\documentclass[journal,10pt]{IEEEtran}

\usepackage{amsmath,amssymb,amsfonts}
\usepackage{color}
\usepackage{pseudocode}
\usepackage{graphicx}
\usepackage{pstricks,pst-all}
\usepackage{algorithmic}
\usepackage{algorithm}

\newtheorem{lemma}{\textbf{Lemma}}
\newtheorem{corollary}{\textbf{Corollary}}
\newtheorem{definition}{\textbf{Definition}}
\newtheorem{theorem}{\textbf{Theorem}}
\newtheorem{observation}{\textbf{Observation}}
\newtheorem{example}{\textbf{Example}}

\newcommand{\ie}{\emph{i.e.}}
\newcommand{\eg}{\emph{e.g.}}
\newcommand{\etal}{~\emph{et al.}~}

\newcommand{\Bs}[1]{\boldsymbol{#1}}
\newcommand{\Fbb}{\mathbb{F}}
\newcommand{\mc}[1]{\mathcal{#1}}

\newcommand{\gaussnum}[3]{{#1 \brack #2}_#3}
\newcommand{\sspan}[1]{\left\langle #1 \right\rangle}
\newcommand{\Prob}[1]{{\ensuremath{\mathbb{P}\left[ #1 \right]} }}

\newcommand{\transpose}{\top}
\newcommand{\bigo}[1]{O\left( #1 \right)}

\DeclareMathOperator{\rank}{rank}

\DeclareMathOperator{\head}{head}
\DeclareMathOperator{\tail}{tail}
\DeclareMathOperator{\In}{In}
\DeclareMathOperator{\Out}{Out}

\title{Subspace Properties of Network Coding\\ and their Applications}

\author{M.~Jafari~Siavoshani,~\IEEEmembership{Student Member,~IEEE}, 
C.~Fragouli,~\IEEEmembership{Member,~IEEE},
S.~N.~Diggavi,~\IEEEmembership{Member,~IEEE}
\thanks{The work of M.~Jafari~Siavoshani and C. Fragouli was supported 
by the Swiss National Science Foundation through Grant PP00P2-128639.}
\thanks{M.~Jafari~Siavoshani and C.~Fragouli are with the School of Computer 
and Communication Sciences, Ecole Polytechnique F\'{e}d\'{e}rale de Lausanne (EPFL), 
Lausanne CH 1015, Switzerland (e-mail: mahdi.jafarisiavoshani@epﬂ.ch; christina.fragouli@epﬂ.ch).}
\thanks{S.~N.~Diggavi was with the Ecole Polytechnique F\'{e}d\'{e}rale de Lausanne
(EPFL), Lausanne CH 1015, Switzerland. He is now with the Department of
Electrical Engineering, University of California, Los Angeles (UCLA), CA
90095 USA (e-mail: suhas@ee.ucla.edu).}
}

\begin{document}
\maketitle
\begin{abstract}
  Systems that employ  network coding for content distribution 
  convey to the receivers  linear combinations of the source
  packets. If we assume randomized network coding, during this process 
   the network nodes collect random subspaces of the space spanned by the source packets.
  We establish several fundamental properties of the random
  subspaces induced in such a system, and show that these subspaces 
   implicitly carry topological information about the
  network and its state that can be passively collected and
  inferred. We leverage this information towards a number of
  applications that are interesting in their own right, such as
  topology inference, bottleneck discovery in peer-to-peer systems and
  locating Byzantine attackers.  We thus argue that, randomized
  network coding, apart from its better known properties for improving
  information delivery rate, can additionally facilitate network management
  and control.
\end{abstract}

\section{Introduction}

Randomized network coding offers a promising technique for 
content distribution systems. In randomized network coding, each
node in the network combines its incoming packets randomly and sends
them to its neighbours \cite{HoKoMeEfShKa2006,ChWuJa2003}. This is
the approach adopted by most practical applications today. For
example, Avalanche, the first implementation of a peer-to-peer (P2P)
system that uses network coding, adopts such a randomized operation
\cite{GkRo2005,GkMiRo2006}. In ad-hoc wireless and sensor networks as
well, most proposed protocols employing network coding again opt for
randomized network operation (see \cite{CWL:Infocom2006} and references therein).

The reason for the popularity of randomized network coding is because it
facilitates a very simple and flexible network operation without need
of synchronization among network nodes, that is well suited to packet
networks. To every packet, a coding vector is appended that determines
how the packet is expressed with respect to the original data packets
produced at the source node. When intermediate nodes combine packets,
the coding vector keeps track of the linear combinations contained in
a particular packet. A receiver which collects enough packets, uses
the coding vectors to determine the set of linear equations it needs
to solve in order to recover the original data packets.

Our contributions start with the observation that coding vectors
implicitly carry information about the network structure as well as
its state\footnote{By state we refer to link or node failures,
congestion in some part of the network, etc.}. Such vectors belong
to appropriately defined vector spaces, and we are interested in
fundamental properties of these (finite-field) vector spaces.  In
particular, since we are investigating properties induced by randomized
network coding, we need to characterize random subspaces of the
aforementioned vector spaces. These properties of random subspaces
over finite fields might be of independent interest. 
We aim to show, using these
properties, that observing the coding vectors we can passively
collect structural and state information about a network. We can
leverage this information towards several applications that are
interesting in their own merit, such as topology inference, network
tomography, and network management (we do not claim here the 
design of practical protocols that use these properties). 
However, we show that randomized network coding, apart from its 
better known properties for facilitating information delivery, 
can provide us with information about the network itself.

To support this claim, we start by studying the problem of passive
topology inference in a content distribution system where intermediate nodes
perform randomized network coding. We show that the subspaces nodes
collect during the dissemination process have a dependence with each
other which is inherited from the network structure. Using this
dependence, we describe the conditions that let us perfectly
reconstruct the topology of a network, if subspaces of all nodes at
some time instant are available.

We then investigate a reverse or dual problem of topology
inference, which is, finding the location of Byzantine attackers. In a
network coded system, the adversarial nodes in the network can disrupt
the normal operation of information flow by inserting erroneous
packets into the network. We use the dependence between subspaces
gathered by network nodes and the topology of the network to extract
information about the location of attackers. We propose several
methods, compare them and investigate the conditions that allow us to
find the location of attackers up to a small uncertainty.

Finally, we then observe that the received subspaces, even at one specific
node, reveal some information about the network, such as the existence
of bottlenecks or congestion. We consider P2P networks for content
distribution that use randomized network coding
techniques. It is known that the performance of such P2P networks
depends critically on the good connectivity of the overlay
topology. Building on our observation, we propose algorithms for
topology management to avoid bottlenecks and clustering in
network-coded P2P systems. The proposed approach is decentralized,
inherently adapts to the network topology, and reduces substantially
the number of topology rewirings that are necessary to maintain a well
connected overlay; moreover, it is integrated in the normal content
distribution.

The paper is organized as follows. We start with the notation and problem
modeling in \S\ref{chp:ProbDefMod}. We investigate the
properties of vector spaces in a system that employs randomized network
coding in \S\ref{chp:NotesVecSpace} and these properties give the framework
to explore applications in \S\ref{chp:TopInfer},
\S\ref{chp:ByzantineAttack}, and \S\ref{chp:TopManagement}. Finally, 
we conclude the paper with a discussion in
\S\ref{chp:conclusion}. Shorter versions of these results have also
appeared in \cite{JaFrDi2007,JaFrDiGk2007,JaFrDi2008}.

\subsection{Related Work}
Network coding started by the work of Ahlswede\etal\cite{AhlCaiLiYeu} 
who showed that a source can multicast
information at a rate approaching the smallest min-cut between the
source and any receiver if the middle nodes in the network combine the
information packets. Li\etal\cite{LiYeCa2003_InfoTheoryTran} showed
that linear network coding with finite field size is sufficient for
multicast. Koetter\etal\cite{KoMe2003_TransNet} presented an
algebraic framework for linear network coding.

Randomized network coding was  proposed by 
Ho\etal\cite{HoKoMeKaEf2003} where they showed that randomly choosing the
network code leads to a valid solution for a multicast problem with
high probability if the field size is large. It was later applied by
Chou\etal\cite{ChWuJa2003} to demonstrate the practical aspects of
random linear network coding. Gkantsidis\etal\cite{GkRo2005,GkMiRo2006} 
 implemented a practical file sharing system based on this idea.
Several other works have also adopted randomized network coding
for content distribution, see for example \cite{add1,add2,add3}.

Network error correcting codes, that are capable of correcting errors
inserted in the network, have been developed during the last few
years. For example see the work of Koetter\etal\cite{KoKs2007},
Jaggi\etal\cite{JaLaKaHoKaMe2007}, Ho\etal\cite{HoLeKoMeEfKa2004}, 
Yeung\etal\cite{YeCa2006_1,CaYe2006_2},
Zhang \cite{Zh2006}, and Silva\etal\cite{SiKsch-IT2011}. 
These schemes are capable of delivering
information despite the presence of Byzantine attacks in the network
or nodes malfunction, as long as the amount of undesired information
is limited. These network error correcting schemes allow to correct 
malicious packet corruption up to certain rate. In contrast, we use
network coding to identify malicious nodes in our work.
Recently, and following our work \cite{JaFrDi2008}, 
 additional approaches are proposed in the literature,
some building on our results \cite{KeLi-InfoCom09}.

Overlay topology monitoring and management that do not employ network
coding has been an intensively studied research topic, see for example
\cite{ResOvNet}. However, in the context of network coding, it is a new
area of research. Fragouli\etal\cite{FrMa2005,FrMaDi2006} took
advantage of  network coding capabilities  for active link loss network
monitoring where the focus was on link loss rate inference. Passive
inference of link loss rates has also been proposed by
Ho\etal\cite{HoLeChWeKo2005}. In a subsequent work of ours, 
Sharma\etal\cite{ShJaDe-ITA07}
study  passive topology estimation for the upstream nodes of every
network node. This work is based on the assumption 
that the local coding vectors for each node
in the network are fixed, generated in advance and known by all other nodes 
in the network, unlike our work that builds on randomized operation.
The idea of passive inference of topological properties from  subspaces 
that are build over time, as far as we know, is a novel contribution of 
this work.

\section{Models: Coding and Network Operation}\label{chp:ProbDefMod}

A simple observation motivates much of the work presented in this
paper:  the subspaces gathered by the network nodes  during information 
dissemination with randomized network coding,  are not  completely random, 
but have some relationship, and this relationship conveys information about 
the network topology as well as its state.   
We will thus  investigate properties of the collected subspaces 
and show how we can use them for diverse applications.

Different properties of the subspaces are relevant to each particular
application and therefore we will develop a framework for investigating
these properties. This will also involve some understanding of modeling
the problem to fit the requirements of an application and then developing
subspace properties relevant to that model.

\subsection{Notation}\label{sec:NotSubSpace}

Let $q\ge 2$ be a power of a prime. In this paper, all vectors and matrices 
have elements in a finite field $\Fbb_q$. We use $\Fbb^{n\times m}_q$ to
denote the set of all $n \times m$ matrices over $\Fbb_q$, and $\Fbb^\ell_q$ to 
denote the set of all row vectors of length $\ell$. The set $\Fbb^\ell_q$ forms 
an $\ell$-dimensional vector space over the field $\Fbb_q$. Note that all vectors 
are row vectors unless otherwise stated. Bold lower-case letters, \eg, $\Bs{v}$, 
are used for vectors and bold capital letters, \eg, $\Bs{X}$, are used to denote matrices.

For a set of vectors $\{\Bs{v}_1,\ldots,\Bs{v}_k\}$ we denote their linear span by $\sspan{\Bs{v}_1,\ldots,\Bs{v}_k}$. For a matrix $\Bs{X}$, $\sspan{\Bs{X}}$
is the subspace spanned by the rows of $\Bs{X}$. We then have $\rank(\Bs{X}) = \dim(\sspan{\Bs{X}})$.

We denote subspaces of a vector space by $\Pi$ and sometimes also
by $\pi$. In this paper, we work on a vector space $\mathbb{F}_q^\ell$ of dimension $\ell$ defined over a finite field $\mathbb{F}_q$. 
For two subspaces
$\Pi_1,\Pi_2\subseteq\mathbb{F}_q^\ell$, we will denote
their intersection by $\Pi_1\cap\Pi_2$ and their joint span 
by $\Pi_1+\Pi_2$ where
\[
\Pi_1 + \Pi_2 \triangleq \{\Bs{v}_1+\Bs{v}_2|\Bs{v}_1\in\Pi_1,\Bs{v}_2\in\Pi_2 \},
\]
is the smallest subspace that contains both $\Pi_1$ and $\Pi_2$.
It is  well known that 
\[
\dim(\Pi_1+\Pi_2)=\dim(\Pi_1)+\dim(\Pi_2)-\dim(\Pi_1\cap\Pi_2).
\]
We also use the following metric to measure
the distance between two subspaces,
\begin{align}\label{eq:SubspaceMetric}
d_S(\Pi_1,\Pi_2) &\triangleq \dim(\Pi_1+\Pi_2)-\dim(\Pi_1\cap\Pi_2)\\
&= \dim(\Pi_1)+\dim(\Pi_2)-2\dim(\Pi_1\cap\Pi_2)\nonumber.
\end{align}
This metric was also introduced in \cite{KoKs2007}, where it was used
to design error correction codes.

In addition to the metric $d_S(\cdot,\cdot)$ defined above,
in some cases we will also need a measure that compares how
a set $\mathcal{A}$ of subspaces differs from
another set $\mathcal{B}$ of subspaces. For this we will
use the average pair-wise distance defined as follows
\begin{equation}\label{distance}
D_S(\mathcal{A},\mathcal{B}) \triangleq \frac{1}{|\mathcal{A}||\mathcal{B}|}\sum_{\pi_a\in \mathcal{A},\pi_b\in \mathcal{B}} d_S(\pi_a,\pi_b). 
\end{equation}
It should be noted that the above relation does not define a metric
for the set of subspaces because the self distance of a set with
itself is not zero. However, $D_S(\cdot,\cdot)$ satisfies the triangle inequality.

In this paper we will be interested in investigating the relationship 
of the collected subspaces at neighboring network nodes.
We consider a network represented as a directed acyclic graph $G=(V,E)$, with
$\vartheta=|V|$ nodes and $\xi=|E|$ edges. For an arbitrary edge
$e=(u,v)\in E$, we denote $\head(e)=v$ and
$\tail(e)=u$. For an arbitrary node $v\in V$, we denote
$\In(v)$ the set of incoming edges to $v$ and
$\Out(v)$ the set of outgoing edges from $v$. 
If a node $u$ has $p$ parents $u_1,\ldots,u_{p}$, we denote with
$P(u)=\{u_1,\ldots,u_{p}\}$ the set of parents of $u$. 
We use $P^l(u)$ to denote the set of all ancestors of $u$ at distance $l$ from $u$ in
the network (we say that two nodes $u$ and $v$ are at distance $l$ if there exists a path of length exactly $l$ that connects them).
We denote with  $\pi^{(u_i)}_u(t)$ the subspace
node~$u$ receives from parent $u_i$ at exactly time $t$,
and  with $\pi_u(t)$  the whole subspace (from all parents) that node $u$
receives at time $t$, that is  $\pi_u(t)=\sum_{i=1}^{p}\pi_u^{(u_i)}(t)$.
We also denote
with $\Pi^{(u_i)}_u(t)$ the subspace node~$u$ has received from
parent $u_i$ up to time $t$, that is, \mbox{$\Pi^{(u_i)}_u(t)=\Pi^{(u_i)}_u(t-1)
  + \pi^{(u_i)}_u(t)$}. Then the subspace $\Pi_u(t)$ that the node has at time $t$ can be expressed as
 \mbox{$\Pi_u(t)=\sum_{i=1}^{p} \Pi^{(u_i)}_u(t)$}.  
For a set of nodes $\mathcal{U}=\{u_1,\ldots,u_{p}\}$, we define
\mbox{$\Pi_\mathcal{U} = \Pi_{u_1}+\cdots + \Pi_{u_{p}}$}.

Finally, we use the big $O$ notation which is defined as follows. Let
$f(x)$ and $g(x)$ be two functions defined on some subset of the real
numbers. We write $f(x)=O(g(x))$
if and only if
there exists a positive real number $M$ and a real number $x_0$ such that
$|f(x)|\le M|g(x)|$ for all $x>x_0$. During the rest of the paper 
we use $O$ to compare functions of the field size $q$, unless otherwise stated.
For example, we will use $f(q)=O(q^{-1})$ to imply that the value of $f(q)$
goes to zero as $q^{-1}$  for $q\rightarrow\infty$.

\subsection{Network Operation}\label{sec:GenNetOperation}

We assume that there is an information source located on a node $S$ 
that has a set of $n$  packets (messages) 
$\{\boldsymbol{x}_1,\ldots,\boldsymbol{x}_n\}$, $\boldsymbol{x}_i\in \mathbb{F}_q^{\ell}$,
to distribute to a set of receivers, where each packet is a sequence 
of $\ell$ symbols over the finite field $\mathbb{F}_q$. To do so, we 
will employ a dissemination protocol based on randomized network 
coding, namely, where each network node sends random linear 
combinations (chosen to be uniform over $\mathbb{F}_q$) of its collected 
packets to its neighbors. We assume for simplicity that there are no packet-losses.

\subsection*{Dissemination Protocol}

It is possible to separate the dissemination protocols into the
following operation categories.
\begin{itemize}
\item {\em Synchronous:} All nodes are synchronized and transmit to
  their neighbors according to a global clock tick (time-slot). At
  timeslot $t\in\mathbb{N}$, node $v$ sends linear combinations from all 
vectors it has  collected up to time \mbox{$t-1$}. Once nodes start transmitting
  information, they keep transmitting until all receivers are able to
  decode. 
\item {\em Asynchronous:} Nodes transmit linear combinations at
  randomly and independently chosen time instants.
\end{itemize}

In this paper, we focus on the
synchronous network where we assume that each link has unit delay\footnote{Unit 
delay can model a buffering window a node needs to wait to collect packets
from all its neighbors.} corresponding to each timeslot, however 
our results can be extended to asynchronous networks as well.

Next, we explain in detail the dissemination protocol, that is
 summarized in
Algorithm~\ref{alg_diss}.

{
\paragraph*{\bf\em Timing} We depict in Fig.~\ref{fig:timing} the relative timing of events within a timeslot. 
Nodes transmit at the beginning of a timeslot. We assume that each packet is received
by its intended receiver before the end of the timeslot. Thus, the timeslot duration incorporates
the packet propagation delay  in one edge of the network. 
}

\begin{figure}[hbt]
\begin{center}
\psset{unit=0.07in}
\begin{pspicture}(2,2)(45,18)
\psset{linewidth=0.5pt}
\psline{->}(2,3)(35,3) \rput(36,1.5){\scriptsize{Time}}

\psline{-}(6,2.5)(6,3) \rput(6,1){\scriptsize{$t-1$}}
\psline[linestyle=dotted,dotsep=2pt](6,3)(6,6)
\psline{-}(26,2.5)(26,3) \rput(26,1){\scriptsize{$t$}}
\psline[linestyle=dotted,dotsep=2pt](26,3)(26,6)
\psline{<->}(6.5,6.1)(25.5,6.1) \rput(16,5.1){\scriptsize{Slot number $t$}}

\rput(36,17){\circlenode{A}{A}}
\rput(36,10){\circlenode{B}{B}}
\ncline{->}{A}{B} 

\psline{<-}(8.8,3)(8.8,17)(25,17) \rput[Cl](25.5,17){\scriptsize{$A$ Transmits}}
\psline{<-}(23.5,3)(23.5,10)(25,10) \rput[Cl](25.5,10){\scriptsize{$B$ Receives}}

\psline{<-}(26.5,3.5)(30,6) \rput[Cl](30.5,6){$\left\{\text{\parbox{2.5cm}{\scriptsize{The point that subspaces are measured: $\Pi_B(t)$}}}  \right.$}
\psline{<-}(5.5,3.5)(3,6)(3,10) \rput[Cb](3,10){\scriptsize{$\Pi_A(t-1)$}}
\end{pspicture}
\end{center}
\caption{Timing schedule of the dissemination protocol given by Algorithm~\ref{alg_diss}.}
\label{fig:timing}
\end{figure}
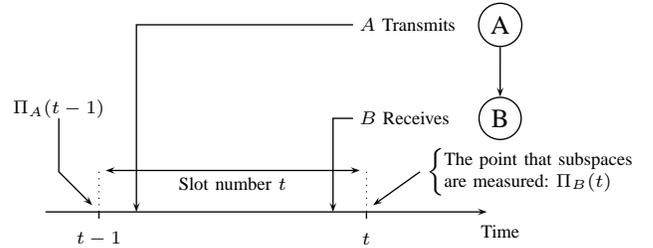

\paragraph*{\bf\em Rate Allocation and Equivalent Network Graph} The dissemination protocol first associates with
each link $e$ of the network a rate $r_e$ (measured as the number of packets transmitted per timeslot on edge $e$). 
These rates are selected in advance using a rate allocation method, for example \cite{Lun}. 

For the rest of the paper, we consider an equivalent network graph, where
each edge $e$ has capacity equal to its allocated rate $r_e$.
On this new graph, we can define the min-cut $c_v$ from the source node $S$ 
to a node $v\in V$. Whenever we refer to min-cut values in the following, we 
imply min-cut values over this equivalent graph. 

We assume that the rate allocation protocol we use satisfies 
\begin{equation} \label{eq_rate} 
r_e\le \min[c_e,c_{\tail(e)}],
\end{equation}
where $c_e$ is the capacity of edge $e$.
This very mild assumption says that the node $v=\tail(e)$ does not send more 
information than it receives, and is satisfied by all protocols that do not 
send redundant packets, \ie, observe flow conservation.

In our work, we consider the case where  $n \gg c_v$, namely,
the dissemination of the $n$ source packets to the receivers 
takes place by using the network over several timeslots.

\paragraph*{\bf\em Node operation} When the dissemination starts, at timeslot say zero, the source 
starts transmitting at each time slot and to each of its outgoing 
edges $e$, $r_e$  randomly selected linear combinations of $n$ 
information packets. We  will call $r_S$ the {\em source rate}.
The source continues until it has transmitted linear combinations 
of all $n$ packets, \ie, for $\frac{n}{r_S}$ timeslots.  
Every other node $v\in V\setminus\{S\}$ in the network, operates as follows:
\begin{itemize}
\item Initially it does not transmit, but only collects in a buffer packets 
from its parents, until a time $\tau_v$, which we call {\em waiting time} 
and we will define in the following. As we will see, each node can decide 
the waiting time by itself and independently from other nodes.
\item At each timeslot $t$,  for all $t\geq \tau_v+1$, it transmits to each 
outgoing edge $e$, $r_e$ linear combinations of all packets it has 
collected in its buffer up to time $t-1$. 
\end{itemize}

\paragraph*{\bf\em Collected Subspaces} 
We can think of each of the $n$ source messages $\{\Bs{x}_i\}$ as corresponding 
to one dimension of an $n$-dimensional space  $\Pi_S\subseteq\mathbb{F}_q^\ell$
where $\Pi_S=\sspan{\Bs{x}_1,\ldots,\Bs{x}_n}$. 
We say that node~$v\in V$ at time $t$ observes a subspace $\Pi_v(t) \subseteq \Pi_S$, with 
dimension $d_v(t)\triangleq\dim(\Pi_v(t))$, if $\Pi_v(t)$
is the space spanned by the received vectors at node $v$ up to time
$t$. Initially, at time $t=0$, the collected subspaces of all nodes (apart the
source) are empty; $d_v(0)=0$, $\forall v\in V\setminus\{S\}$.

\paragraph*{\bf\em Waiting Times}
We next define the waiting times, that will be used in
the following sections to ensure that the subspaces of different
nodes be distinct, and are a usual assumption in dissemination
protocols; indeed, for   large $n$ the waiting time does not
affect the rate. 
For example, in the information-theoretic proof of the
main theorem in network coding \cite{AhlCaiLiYeu}, each node waits
until it collects at least one message from each of its incoming links
before starting transmissions. 

\begin{definition}\label{def:waiting_time}
The {\em waiting time} $\tau_v$ for a node $v$ is the first timeslot 
during which node $v$ receives information from the source at a rate 
equal to its min-cut $c_v$, and additionally, has collected in its 
buffer a  subspace of dimension at least $c_v+1$.
\end{definition} 

Note that, because we are dealing with acyclic graphs, we can impose a 
partial order on the waiting times of the nodes, such that all parents 
of a node have smaller waiting time than the node. Moreover, each node 
can decide whether the conditions for the waiting time are met, by observing 
whether it receives information at a rate equal to its min-cut, and what 
is the dimension of the subspace it has collected. That is, a node does 
not need to know any topological information (apart from its min-cut), and 
the waiting times do not need to be communicated in advance to the nodes, but 
can be decided online based on the network conditions.

\begin{center}
\begin{pseudocode}[doublebox]{DisseminationProtocol}{G=(V,E),S,n,\tau_v,r_e }
\FOREACH v \in V\setminus \{S\} \DO \Pi_v(0)=\varnothing,d_v(0)=0\\
t \GETS 1\\
\WHILE \min_v{d_v(t)}<n \\
\BEGIN
\FOREACH v \in V \\
\BEGIN
\IF t\geq \tau_v +1 \THEN 
\BEGIN
\FOREACH e\in \Out(v) \DO
\BEGIN
\text{node $v$ transmits from}\\ \text{$\Pi_v(t-1)$ with rate $r_e$ on $e$}\\
\END
\END
\END\\
\FOREACH v \in V \DO \;\text{update}\; \Pi_v(t),\; d_v(t) \\
t \GETS t+1\\
\END
\label{alg_diss}
\end{pseudocode}
\\
\vspace{-0.3cm}
Alg.~\ref{alg_diss}: Dissemination protocol.
\end{center}

\subsection*{Source Operation and the Source Subspace  $\Pi_S$}

As we discussed, the source needs to convey to the receivers $n$ source packets
that span the $n$-dimensional subspace  $\Pi_S=\sspan{\Bs{x}_1,\ldots,\Bs{x}_n}$,
with $\Pi_S\subseteq\mathbb{F}_q^\ell$. 
$\Pi_S$ is isomorphic to $\Fbb_q^n$; thus, for the purpose of studying relationships between subspaces of $\Pi_S$,
we can equivalently assume that $\Pi_S=\Fbb_q^n$, and that
 node~$v\in V$ at time $t$ observes a subspace $\Pi_v(t) \subseteq \Pi_S$.
This simplification is very natural in the case where we employ coding 
vectors, reviewed briefly in the following, as we only need consider the 
coding vectors for our purposes and ignore the remaining contents of the 
packets; however, we can also use the same approach in the case where the 
source performs noncoherent coding, described subsequently.

\paragraph{Use of Coding Vectors}
\label{subsec:NetCodCodVec}
To enable receivers to decode, the source assigns $n$ symbols of each
message vector (packet) to determine the linear relation between that packet
and the original vectors $\Bs{x}_i$, $i=1,\ldots,n$. Without loss of
generality, let us assume these $n$ symbols (which form a vector of
length $n$) are placed at the  beginning of each message vector. This
vector is called \emph{coding vector}. Each message vector $\Bs{x}_i$
contains two parts. The vector $\Bs{x}_i^C\in\mathbb{F}_q^n$ with length
$n$ is the coding vector and remaining part,
$\Bs{x}_i^I\in\mathbb{F}_q^{\ell-n}$, is the information part where
\[
\Bs{x}_i\triangleq [\Bs{x}_i^{C}\ |\ \Bs{x}_i^{I}].
\]
The coding vectors $\Bs{x}_i^C$, $i=1,\ldots,n$ are chosen such that they
form a basis for $\mathbb{F}_q^n$. For simplicity we assume
$\Bs{x}_i^C=\Bs{e}_i$ where $\Bs{e}_i\in\mathbb{F}_q^n$ is a vector with one at
position $i$ and zero elsewhere.

For our purposes, it is sufficient to restrict our algorithms to examine the coding vectors. Thus, the source has the space $\Pi_S=\mathbb{F}_q^n$; during the information dissemination,
if a node $v$  at time $t$  has collected $m$ packets
$\Bs{z}_i$ with coding vectors $\Bs{z}_i^C$, it has observed the subspace
$\Pi_v(t)=\sspan{\Bs{z}_1^C,\ldots,\Bs{z}_m^C}$. In other words, the coding vectors capture all the information we need for our applications.

\paragraph{Subspace Coding}\label{sec:Mod:ObliviousModel}

Our approach also works in the case of subspace coding,  that was introduced in
\cite{KoKs2007}. 
We next briefly explain the idea of communication using subspaces, in a
network performing randomized network coding. 

In the following, we use the same notation as
introduced in \S\ref{sec:GenNetOperation}. Let $\{\Bs{x}_1,\ldots,\Bs{x}_n\}$,
$\Bs{x}_i\in\mathbb{F}_q^\ell$ denote the set of packets the source
has. Assume that there is no error in the network. An arbitrary
receiver $R_v$ at node $v$ collects $m$ packets $\Bs{z}_i$, $i=1,\ldots,m$,
where each $\Bs{z}_i$ can be presented as $\Bs{z}_i=\sum_{j=1}^n h_{ij} \Bs{x}_j$. The
coefficients $h_{ij}$ are unknown and randomly chosen over
$\mathbb{F}_q$. 
In matrix form, the transmission model can be represented as
\[
\Bs{Z}_v=\Bs{H}_{Sv} \Bs{X},
\]
where $\Bs{H}_{Sv}\in\mathbb{F}_q^{m\times n}$ is a random matrix and
$\Bs{X}\in\mathbb{F}_q^{n\times\ell}$ is the matrix whose rows are the
sources' packets.

The matrices $\Bs{H}_{Sv}$ are randomly chosen, under constraints
imposed by the network topology. As stated in \cite{KoKs2007} and
proved in \cite{JaFrDi2008-ISIT,SiKschKo-CapRandNetCod,JaMoFrDi-IT11}, the above model naturally leads 
to consider information transmission not via the choice of $\Bs{x}_i$ but
rather by the choice of the vector space spanned by $\{\Bs{x}_i\}$.

In the case of subspace coding, the dissemination algorithm works in exactly the same way as in the case of coding vectors; what changes is how the source maps the information to the packets it transmits, and how decoding occurs. However, this is orthogonal to our purposes, since we perform no decoding of the information messages, but simply observe the relationship between the subspaces different nodes in the network collect. Thus, the same approach can be applied in this case as well.

\subsection{Input to Algorithms}

We are interested in designing algorithms that leverage the relationships between subspaces observed at different network nodes for network management and control.
The algorithms design will depend on the
information that we have access to. We distinguish between the following.
\begin{itemize}
\item {\em Global information}: A central entity knows the subspaces
  that all $\vartheta$ nodes in the network have observed.
\item {\em Local Information:} There is no such omniscient entity, and
  each node $v$ only knows what it has received, its own subspace
  $\Pi_v$.
\end{itemize} 
We may also have information between these two extreme cases. Moreover,
we may have a {\em static view}, where we take a snapshot of the
network at a given time instant $t$, or a {\em non-static view}, where
we take several snapshots of the network and use the subspaces'
evolution to design an algorithm.

We will argue in Section \ref{chp:TopInfer} that capturing even global information can be accomplished with relatively low overhead (sending one additional packet per node at the end of the dissemination protocol); thus, the algorithms we develop even assuming global information can in fact be implemented almost passively and at low cost.

\section{Properties of Random Vector Spaces over a Finite Field $\mathbb{F}_q^n$}
\label{chp:NotesVecSpace}

In this section, we will state and prove basic properties and results that we will
exploit towards various applications in the following sections. In
particular, we will investigate the properties of random sampling from
vector spaces over a finite field. 
Such properties 
give us a better insight and understanding of randomized
network coding and form a foundation 
for the results and algorithms presented in this paper.

\subsection{Sampling Subspaces over $\mathbb{F}_q^n$}
Here, we explore properties of randomly sampled
subspaces from a vector space $\mathbb{F}_q^n$. 
We start with the following lemma that explores properties of a single subspace.

\begin{lemma}\label{lem:rand_choose_full_rank} 
Suppose we choose $m$ vectors from an $n$-dimensional vector space
$\Pi_S = \mathbb{F}_q^n$ uniformly at random to construct a space
$\Pi$. Then the subspace $\Pi$ will be
full rank (has dimension $\min[m,n]$) w.h.p. (with high
probability)\footnote{Throughout this paper, when we talk about an
event occurring with high probability, we mean that its
probability behaves like $1-\bigo{q^{-1}}$, which goes to
$1$ as $q\rightarrow\infty$.}, namely,
\begin{equation*}
\Prob{\dim(\Pi)=\min[m,n]} = [1-\bigo{q^{-1}}]. 
\end{equation*}
\end{lemma}
\begin{proof}
Refer to Appendix~\ref{apndx:Proofs}.
\end{proof}

We conclude that for large values of $q$, selecting $m\le n$ vectors uniformly at random
from $\mathbb{F}_q^n$ to construct a subspace $\Pi$ is equivalent to
choosing an $m$-dimensional subspace from $\mathbb{F}_q^n$ uniformly at random.
Note that this is not true for small values of $q$.

We next examine connections between multiple subspaces.
\begin{lemma}\label{lem:subset}
Let $\Pi_1$ and $\Pi_2$ be two subspaces of $\Pi_S=\mathbb{F}_q^n$ with
dimension $d_1$ and $d_2$ respectively, intersection of dimension
$d_{12}$ and  $\Pi_1\nsubseteq\Pi_2$ (\ie, $d_{12}<d_1$). 
Construct $\Pi_1'$ by choosing $m$ vectors from $\Pi_1$
uniformly at random. Then \mbox{$\Prob{\Pi_1' \subset
\Pi_2}=\bigo{q^{-m}}.$}
\end{lemma}
\begin{proof}
Refer to Appendix~\ref{apndx:Proofs}. 
\end{proof}

\begin{lemma}\label{lem:JointSpacePi_k}
  Suppose $\Pi_k$ is a $k$-dimensional subspace of a vector space
  $\Pi_S=\mathbb{F}_q^n$. Select $m$ vectors uniformly at random from
  $\Pi_S$ to construct the subspace $\Pi$. We have
\begin{eqnarray}
\dim(\Pi\cap\Pi_k)&=& \min[k,(m-(n-k))^+] \nonumber\\
&=&\left(\min[m,n]+k-n\right)^+,
\end{eqnarray}
with probability $1-\bigo{q^{-1}}$. 
\end{lemma}
\begin{proof}
Refer to Appendix~\ref{apndx:Proofs}. 
\end{proof}

\begin{corollary}\label{cor:SpaceJointDim}
  Suppose $\Pi_1$ and $\Pi_2$ are two subspace of $\mathbb{F}_q^n$
  with dimension $d_1$ and $d_2$ respectively and joint dimension
  $d_{12}$. Let us take $m_1$ vectors uniformly at random from $\Pi_1$
  and $m_2$ vectors from $\Pi_2$ to construct subspaces $\hat{\Pi}_1$
  and $\hat{\Pi}_2$. We have
\begin{align*}
\dim(\hat{\Pi}_1\cap\hat{\Pi}_2)=& \min \left[d_{12},(m_1+m_2-(d_1+d_2-d_{12}))^+,\right.\\
&\left. (m_1-(d_1-d_{12}))^+, (m_2-(d_2-d_{12}))^+\right],
\end{align*}
with probability $1-\bigo{q^{-1}}$.
\end{corollary}
\begin{proof}
Refer to Appendix~\ref{apndx:Proofs}.
\end{proof}

By choosing $\Pi_1=\Pi_2=\mathbb{F}_q^n$ in Corollary~\ref{cor:SpaceJointDim} we have the following corollary.
\begin{corollary}
  Let us construct two subspaces $\hat{\Pi}_1$ and $\hat{\Pi}_2$ by
  choosing $m_1$ and $m_2$ vectors uniformly at random respectively
  from $\mathbb{F}_q^n$. Then the subspaces $\hat{\Pi}_1$ and
  $\hat{\Pi}_2$ will be disjoint with probability
  $1-\bigo{q^{-1}}$ if $m_1+m_2\le n$.
\end{corollary}

We are now ready to discuss one of the important properties 
of randomly chosen subspaces which is very useful for our work: randomly
selected subspaces tend to be ``as far as possible''. We will clarify
and make precise what we mean by ``as far as possible'', see also  \cite{BaFr}. We first review the
definition of a subspace in general position with respect to a family of
subspaces.
\begin{definition}[{\cite[Chapter~3]{BaFr}}]
  Let $\Pi_S$ be an $n$-dimensional space over the field $\mathbb{F}_q$ and for
  $i=1,\ldots,r$, let $\Pi_i$ be a subspace of $\Pi_S$,
 with $\dim(\Pi_i)=d_i$. A subspace $\Pi\subseteq \Pi_S$ of dimension $d$ is
  in general position with respect to the family $\{\Pi_i\}$ if
\begin{equation}\label{equ:genpos}
\dim(\Pi_i\cap \Pi)=\max\left[d_i+d-n,0 \right],\quad \forall i\in\{1,\ldots,r\}.
\end{equation}
\end{definition}

It should be noted that $\max[d_i+d-n,0]$ is the minimum possible
dimension of $(\Pi_i\cap \Pi)$. So what the above definition says is
that the intersection of $\Pi$ and each $\Pi_i$ is as small as
possible. Using the above definition we can state the following
theorem\footnote{Versions of this theorem can be easily derived from results in the literature \cite{BaFr}, 
but we repeat here the short derivation for completeness}.
\begin{theorem}\label{thm:SubSpaceGenPos}
  Suppose $\{\Pi_i\}$, $i=1,\ldots,r$, are subspaces of
  $\Pi_S=\mathbb{F}_q^n$. Let us construct a subspace $\Pi$ by randomly
  choosing $m$ vectors from $\Pi_S$. Then $\Pi$ will be in general
  position with respect to the family $\{\Pi_i\}$ 
  w.h.p., \ie, with probability $1-\bigo{q^{-1}}$.
\end{theorem}
\begin{proof}
Refer to Appendix~\ref{apndx:Proofs}. 
\end{proof}

Theorem~\ref{thm:SubSpaceGenPos} demonstrates a nice property of
randomized network coding where the subspaces spanned by coding
vectors tend to be as far as possible on different paths of the 
network.

\subsection{Rate of Innovative Packets}
In the following sections, we will need to know the rate of receiving
innovative message vectors (packets) at receivers in a dissemination protocol performing
randomized network coding. By innovative we refer to vectors that do not belong in the space spanned by already collected packets.
As it is shown in \cite{AhlCaiLiYeu}, the
source can multicast at rate equal to the minimum min-cut of all
receivers if the intermediate nodes can combine the incoming messages. 
Moreover, it is shown in \cite{LiYeCa2003_InfoTheoryTran} that using 
linear combinations is sufficient to achieve information transfer at a rate equal to the minimum mincut of all receivers.
In \cite{AhlCaiLiYeu,HoKoMeEfShKa2006}, it is also demonstrated that choosing the coefficients 
of the linear combinations randomly is sufficient (no network-specific code design is required) with high 
probability if the field size is large enough.

To find the rate of receiving information at each node where the
implemented dissemination protocol performs randomized network 
coding, we can use the following result given in Theorem~\ref{thm:RandNetCodRate}.
Note that our described dissemination protocol, although 
very common in practice, does not exactly fit to the 
previous theoretical results in the literature that 
examine rates, because  the operation of the network nodes is
not memory-less. That is, while for example in 
\cite{HoKoMeEfShKa2006,AhlCaiLiYeu,LiYeCa2003_InfoTheoryTran} 
each transmitted packet at time $t$ is a function of a 
small subset of the received packets up to time $t$ (the 
ones corresponding to the same information message), in our 
case a packet transmitted at time $t$ is a random linear 
combination of all packets received up to time $t$.
This small variant of the main theorem on randomized 
network coding is very intuitive, and we formally state it in following.
\begin{theorem}\label{thm:RandNetCodRate}
Consider a source that transmits  $n$ packets over a connected network using
the dissemination protocol described in \S\ref{sec:GenNetOperation},
and assume that the network nodes perform random linear network coding
over a sufficiently large finite field.
Then there exists $t_0$ such that for all $t>t_0$ 
each node $v$ in the network receives $c_v$   independent linear combinations of the $n$ source packets per time slot, where  $c_v=\textnormal{mincut}(v)$.
\end{theorem}
\begin{proof}
Refer to Appendix~\ref{apndx:ProofThmRandNetCodeRate}.
\end{proof}

Given Theorem~\ref{thm:RandNetCodRate}, we can state the following definition.
\begin{definition}
  For a specific information dissemination protocol over a network, 
  we define the \emph{steady state}  as the time period during which each node $v$ in the network  receives exactly $c_v$ independent linear combinations of the $n$ source packets per time slot and none of
  the nodes, except source $S$, has collected $n$ linearly independent
  combinations. 
We call the  time that the network enters  steady state phase
  the \emph{steady state starting time} and denote it by $T_s$.
If the network never attains the steady state phase then we use
$T_s=\infty$.
\end{definition}

For our protocol in \S\ref{sec:GenNetOperation}, $T_s$ depends not only 
on the network topology, but also on the waiting times $\tau_v$. For
the waiting time defined in Definition~\ref{def:waiting_time} we can
upper bound $T_s$ as stated in Lemma~\ref{lem:UpperBound_Ts}.
\begin{lemma}\label{lem:UpperBound_Ts}
If $n$ is large enough, for the dissemination protocol given in \S\ref{sec:GenNetOperation}
we may upper bound the steady state starting time as follows
\[
T_s\le 2 D(G)-1,
\]
where $D(G)$ is the longest path from the source to other nodes
in the network\footnote{Note that $D(G)$ is different from the longest
shortest path which is called diameter of $G$ in the graph theory literature.}.
\end{lemma}
\begin{proof}
Refer to Appendix~\ref{apndx:Proofs}.
\end{proof}

In order to be sure that the dissemination protocol given in \S\ref{sec:GenNetOperation} enters the
steady state phase, $n$ should be large enough. Using Lemma~\ref{lem:UpperBound_Ts}
we have the following result, Corollary~\ref{cor:SuffCondon_n}.

\begin{corollary}\label{cor:SuffCondon_n}
A sufficient condition for $n$ to be sure that the protocol enters 
the steady state is that 
\[
2 D(G)-1 < \lfloor \frac{n}{c_{\text{max}}} \rfloor,
\]
where $c_{\text{max}}=\max_{v\in V} c_v$.
\end{corollary}

\section{Topology Inference}\label{chp:TopInfer}

In this section, we will use the tools developed in
\S\ref{chp:NotesVecSpace} to investigate the relation between
the network topology and the subspaces collected at the nodes during
information dissemination. We will develop conditions that allow us
to passively infer the network topology with 
(asymptotically on the value of $q$) no error.  
The proposed scheme is passive in the sense that it does not alter the
normal data flow of the network, and the information rates that can be achieved.
In fact, we can think of our protocol as identifying the topology
of the network which is induced by the traffic.

We build our intuition starting from information dissemination in tree 
topologies, and then extend our results in arbitrary topologies. Note that
information dissemination using network coding in tree topologies does
not offer throughput benefits as compared to routing; however, it
is an interesting case study that will naturally lead to our framework
for general topologies.
We then define conditions under which the topology of a tree and that of an
arbitrary network can be uniquely identified using the observed
subspaces. Note that uniquely identifying the topology is a strong
requirement, as the number of topologies for a given number of network
nodes is exponential in the number of nodes.

\subsection{Tree Topologies}\label{sec_tree}
Let $G=(V,E)$ be a network that is a directed tree of
depth $D(G)$, rooted at
the source node $S$. 
We will present {\sf (i)} necessary and
sufficient conditions under which the tree topology can be uniquely
identified, and {\sf (ii)} given that these conditions are satisfied,
algorithms that allow us to do so.

We first consider trees where each edge is allocated the same rate $c$, and
thus the min-cut from the source to each node of the tree equals
$c$. We then briefly discuss the case of undirected trees. Finally we
examine the case where edges are allocated different rates, and thus nodes
may have different min-cuts from the source.

\subsubsection{Common Min-Cut} 

Assume that each edge of the tree has the same capacity $c$ (\ie, a rate 
allocation algorithm has assigned the same rate $r_e=c$ on each edge of the tree).
Thus all nodes in the tree have the same min-cut, equal to $c$.
Then according to the
dissemination protocol introduced in Algorithm~\ref{alg_diss},
each node $v$ will wait time $\tau_v$, until it has collected a $c+1$ dimensional subspace, and then start transmitting to its children.
Our claim is that, we can then identify
the network topology using a single snapshot of all node's subspaces
at a time $t$. 
Before formally proving the result in Theorem~\ref{th_1},
we will give some intuition on why this is so, and why the waiting time is crucial to achieve this. 
We start from an example on the simple  network in Figure~\ref{fig_tree}.

\begin{example}
  Consider the tree  in Figure~\ref{fig_tree}
and assume that the edges have unit capacity ($c=1$). 
 Algorithm~\ref{alg_diss} works as follows.
At time $t=1$, node $A$ receives a vector $y_1$ from the source $S$. Node $A$ waits, as it has not yet collected a $c+1=2$ dimensional subspace.
At time $t=2$, it receives a vector $y_2$. It now has collected the subspace $\Pi_A(2)=\sspan{y_1,y_2}$, and thus at the next timeslot it will start transmitting.
At time $t=3$, node $A$ transmits vectors $y_1^B$ and  $y_1^C$ to nodes $B$ and $C$ respectively, with  $y_1^B,y_1^C \in \Pi_A(2)$. Thus $\Pi_B(3)=\sspan{y_1^B}$ and  $\Pi_C(3)=\sspan{y_1^C}$.
Node $A$ also receives
a vector $y_3$ from the source, and thus $\Pi_A(3)=\sspan{y_1,y_2,y_3}$. 
Consider now the subspaces $\Pi_A(3)$, $\Pi_B(3)$ and $\Pi_C(3)$. 
We see that  $\Pi_B(3)\subseteq \Pi_A(3)$, and  $\Pi_C(3)\subseteq \Pi_A(3)$;
we thus conclude that nodes $B$ and $C$ are children of node $A$. Moreover,
 $\Pi_B(3)\neq \Pi_C(3)$, which will allow us to distinguish between children of these two nodes when we deal with larger trees.

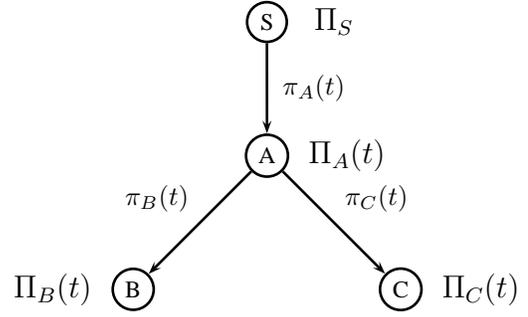
\begin{figure}[hbt]
\begin{center}
\psset{unit=0.07in}
\begin{pspicture}(-3,-3)(23,23)
\psset{linewidth=1.0pt}
\begin{small}
\rput(10,20){\circlenode{S}{S}}\rput(15,20){\large $\Pi_S$}
\rput(10,10){\circlenode{A}{A}}\rput(16,10){\large $\Pi_A(t)$}
\rput(0,0){\circlenode{B}{B}}\rput(-6,0){\large $\Pi_B(t)$}
\rput(20,0){\circlenode{C}{C}}\rput(26,0){\large $\Pi_C(t)$}
\end{small}
\ncline{->}{S}{A}\Aput{$\pi_A(t)$}
\ncline{->}{A}{B}\Bput{$\pi_B(t)$}
\ncline{->}{A}{C}\Aput{$\pi_C(t)$}
\end{pspicture}
\end{center}
\caption{Directed tree with four nodes rooted at the source $S$.}
\label{fig_tree}
\end{figure}

In contrast, if Algorithm~\ref{alg_diss} did not impose a waiting time, and
  node $A$ started transmitting to nodes $B$ and $C$ at time $t=2$,
  then both nodes $B$ and $C$ would receive the same vector $y_1$, {\it
    i.e.}, $\Pi_B(2)=\Pi_C(2)=\sspan{y_1}$. In fact, at all subsequent
  times, we will have that $\Pi_B(t)=\Pi_C(t)=\Pi_{A}(t-1)$.
  Thus, we would not be able to distinguish between these two nodes.

$\hfill\blacksquare$
\end{example}

The main idea in our result is that, if we consider two nodes $u$ and $v$ 
at the network which have collected subspaces $\Pi_u(t)$ and $\Pi_v(t)$ at 
time $t$, then, unless $u$ and $v$ have a child-ancestor relationship (\ie, are 
on the same branch in the tree), it holds that $\Pi_u(t)\nsubseteq\Pi_v(t)$ 
and $\Pi_v(t)\nsubseteq\Pi_u(t)$.

The challenge in proving this is that we deal with subspaces evolving over 
time, and thus we cannot directly apply the results in \S\ref{chp:NotesVecSpace}. 
For example, for the network in  Figure~\ref{fig_tree}, $\Pi_B(t)$ and 
$\Pi_C(t)$ are not subspaces that are selected uniformly at random 
from $\Pi_A(t)$; instead, they are build over time as $\Pi_A(t)$ also 
evolves. We will thus need the following two results, that modify the 
results in  \S\ref{chp:NotesVecSpace} to take into account the time 
evolution in the creation of the subspaces.  We start by examining in 
Lemma~\ref{lem:ConsecutiveSubspaceSampling} the relationship between 
subspaces collected at the immediate children of a given parent node (for 
example, at the children  $B$ and $C$ of node $A$). These are created 
by sampling the same subspaces (those at node $A$). We then examine in 
Corollary \ref{cor:ConsecutiveSubspaceSampling} the relationship between 
subspaces collected at nodes that have different parents (for example, a 
node that has $B$ as parent and a node that has $C$ as parent).

\begin{lemma}\label{lem:ConsecutiveSubspaceSampling}
Suppose there exist (proper) subspaces $\Pi(0)\subset\Pi(1)\subset\cdots\subset\Pi(t-1)$
with dimensions $d_0,\ldots,d_{t-1},$ respectively. Let us construct the set of subspaces
$\Pi_u(i)$, $i=1,\ldots,t$, as follows. Set $\Pi_u(i)=\sum_{j=1}^i \pi_u(j)$ where 
$\pi_u(j)$ is the span of $k_u(j)$ vectors chosen uniformly at random from $\Pi(j-1)$
such that $k_u(1)<d_0$ and $k_u(j)\le (d_{j-1}-d_{j-2})$ for $j=2,\ldots,t$. Similarly,
we construct the set of subspaces $\Pi_v(i)=\sum_{j=1}^i \pi_v(j)$ where for $k_v(j)$
we have similar conditions, namely, $k_v(1)<d_0$ and $k_v(j)\le (d_{j-1}-d_{j-2})$ for 
$j=2,\ldots,t$. Then we have
\[
\Pi_u(i)\nsubseteq \Pi_v(j) \quad \text{and}\quad \Pi_v(j)\nsubseteq \Pi_u(i) \quad \forall i,j\in\{1,\ldots,t\},
\]
with high probability.
\end{lemma}
\begin{proof}
Refer to Appendix~\ref{apndx:Proofs}.
\end{proof}

\begin{corollary}\label{cor:ConsecutiveSubspaceSampling}
Suppose that there exist two set of subspaces $\{\Pi_u(i)\}_{i=0}^{t-1}$ and $\{\Pi_v(i)\}_{i=0}^{t-1}$
such that $\Pi_u(0)\subset\cdots\subset\Pi_u(t-1)$ and $\Pi_v(0)\subset\cdots\subset\Pi_v(t-1)$.
Moreover, assume that $\Pi_u(i)\nsubseteq \Pi_v(j)$ and $\Pi_v(j)\nsubseteq \Pi_u(i)$ 
$\forall i,j\in\{0,\ldots,t-1\}$. Now, construct two set of subspaces $\{\Pi_a(i)\}_{i=1}^t$
and $\{\Pi_b(i)\}_{i=1}^t$ by setting $\Pi_a(i)=\sum_{j=1}^i \pi_a(j)$ and $\Pi_b(i)=\sum_{j=1}^i \pi_b(j)$
where $\pi_a(i)$ is chosen uniformly at random from $\Pi_u(i-1)$ and 
$\pi_b(i)$ is chosen uniformly at random from $\Pi_v(i-1)$ (with some arbitrary dimension).
Then we have
\[
\Pi_a(i)\nsubseteq \Pi_b(j) \quad \text{and}\quad \Pi_b(j)\nsubseteq \Pi_a(i) \quad \forall i,j\in\{1,\ldots,t\},
\]
with high probability.
\end{corollary}
\begin{proof}
Refer to Appendix~\ref{apndx:Proofs}.
\end{proof}

\begin{theorem}\label{th_1}
  Consider a tree of depth $D(G)$ where each edge has capacity
  $c$, and the dissemination Algorithm~\ref{alg_diss}. A static global
  view of the network at time $t$, with
  \mbox{$2 D(G)-1 < t < \lfloor\frac{n}{c} \rfloor$}, allows to uniquely
  determine the tree structure with high probability, if the waiting times
  are chosen according to Definition~\ref{def:waiting_time}.
\end{theorem}

\begin{proof}
We will say that a node of the tree is at level $l$ if it has distance $l$ 
from the source. In a tree there exists a unique path
\mbox{$\mathcal{P}_u=\{S,P^{l_u-1}(u),\ldots,P(u),u\}$} from source $S$ to
node $u$ at level $l_u$ of the network.

If we consider a time $t$ in steady state (where all nodes have nonempty 
subspaces and none has collected the whole space), then clearly using 
Algorithm~\ref{alg_diss} for dissemination in the network for the nodes 
along the path $\mathcal{P}_u$ it holds that
\begin{equation} \label{eq_subsp}
\Pi_u(t)\subset\Pi_{P(u)}(t)\subset\cdots\subset\Pi_{P^{l_u-1}(u)}(t)\subset \Pi_{S}.
\end{equation} 
Note that the conditions on $t$ ensure that the network is in steady-state.

To identify the topology of the tree it is sufficient to show that
$\Pi_u(t)\nsubseteq\Pi_v(t)$ for any node $v$ that is not in
$\mathcal{P}_u$. Let $l_u$ and $l_v$ be the distance of $u$ and $v$ 
from the source, respectively.

First, we observe that, starting from the source,
by applying Lemma~\ref{lem:ConsecutiveSubspaceSampling} and
Corollary~\ref{cor:ConsecutiveSubspaceSampling} and because of 
Definition~\ref{def:waiting_time} the subspaces of the nodes
at the same level (same distance from the source) are different
at all times. So it only remains to check the condition 
$\Pi_u(t)\nsubseteq\Pi_v(t)$ for those node $v$ that are not in the 
same level as $u$.

Consider two cases. First, if $l_u<l_v$ then let $v'$ be the ancestor of $v$ at the same
level as $u$. By Corollary~\ref{cor:ConsecutiveSubspaceSampling} we have
$\Pi_u(t)\nsubseteq\Pi_{v'}(t)$ so
$\Pi_u(t)\nsubseteq\Pi_v(t)$ because $\Pi_v(t)\subseteq\Pi_{v'}(t)$.

Now consider the second case, $l_u>l_v$. We start by assuming $\Pi_u(t)\subseteq\Pi_v(t)$ 
and then we will show that this assumption leads to a contradiction. 
Let $u'$ be the ancestor of $u$ at the same level of $v$. Then we make
the following observation. If at time $t$ we have $\Pi_u(t)\subseteq\Pi_v(t)$
by Lemma~\ref{lem:subset} we should have had $\Pi_{P(u)}(t-1)\subseteq\Pi_v(t)$ and 
so $\Pi_{P^2(u)}(t-2)\subseteq\Pi_v(t)$ and finally we should had had
$\Pi_{u'}(t-l_u+l_v)\subseteq\Pi_v(t)$. But according to 
Corollary~\ref{cor:ConsecutiveSubspaceSampling} this is a contradiction
because $u'$ and $v$ are at the same level.

In the above argument, we have shown that $\Pi_{P(u)}(t)$ 
is the smallest subspace contains $\Pi_u(t)$ among all nodes' subspaces 
at time $t$. So we are done.
\end{proof}

Assume now that Theorem~\ref{th_1} holds. To determine the tree
structure, it is sufficient to determine the unique parent each node
has. From the previous arguments, the parent of node $u$ is the unique
node $v$ such that $\Pi_v(t)$ is the minimum dimension subspace that
contains $\Pi_u(t)$. Then, the parent of node $u$ is the node $v$ such
that 
\[
v= \underset{w\in V:\ d_{uw}=d_u}{\operatorname{arg\,min}} d_w.
\]
As we will discuss in Section~\ref{sec_practical}, 
collecting the subspace information from the network nodes can be implemented efficiently.
The algorithm that determines the tree topology reduces this information to only two ``sufficient statistics'': the
dimension of each subspace $d_u=\dim(\Pi_u), \ \forall u\in V,$ and the
dimension of the intersection of every two subspaces
$d_{uv}=\dim(\Pi_u\cap\Pi_v),\ \forall u,v\in V$, as described in
Algorithm~\ref{alg_ind_tree}, assuming that the conditions of
Theorem~\ref{th_1} hold. 
\begin{center}
\begin{pseudocode}[doublebox]{Tree}{\{d_u\},\{d_{uv}\}}
\FOREACH u \in V \DO 
\BEGIN
\IF d_u=n \THEN  u\GETS S
\ELSE
\BEGIN
\text{node $u$ has parent the node $v$ with}\\
v= \underset{w\in V:\ d_{uw}=d_u}{\operatorname{arg\,min}} d_w
\END
\END
\label{alg_ind_tree}
\end{pseudocode}
\\
\vspace{-0.3cm}
Alg.~\ref{alg_ind_tree}: Find the network topology for a tree.
\end{center}

\subsubsection{Directed v.s. Undirected Network}
In a tree with a single source, since new information can only flow
from the source to each node along a single path, whether the network
is directed or undirected makes no difference.  In other words, from 
\eqref{eq_subsp}, all vectors that a node will send to its
predecessor will belong in the subspace the predecessor already
has. Thus Theorem~\ref{th_1} still holds for undirected networks with
a common mincut.

\subsubsection{Different Min-Cuts}
Assume now that the edges of the tree have 
different capacities, i.e., assigned different rates.
In this case, the proof of 
Theorem~\ref{th_1} still holds, provided that 
the condition in Theorem~\ref{th_1} is modified to 
\[
2D(G)-1 < t < \lfloor\frac{n}{c_{\text{max}}} \rfloor,
\]
where $c_{\text{max}}=\max_{v\in V} c_v$.

We underline that this theorem would not hold without the assumption in (\ref{eq_rate}) .
 Without this condition, it is possible that we cannot distinguish between nodes at
 same level with a common parent as explained in the following example.
 \begin{example}
   If in the network in Figure~\ref{fig_tree}, edge $SA$ has unit
   capacity, while edge $AB$ and $AC$ have capacity two. In this case
   it is easy to see that there exists $t_0$ such that
   \mbox{$\Pi_B(t)=\Pi_C(t)=\Pi_A(t-1)$}, $\forall t\geq t_0$. Clearly
   in this case, we cannot distinguish between nodes $B$ and $C$ with
   this dissemination protocol.
 $\hfill\blacksquare$
 \end{example}

\subsection{General Topologies}
\label{sec_general}
Consider now an arbitrary network topology, corresponding to a
directed acyclic graph. 
An intuition we can get from examining tree structures
is that, we can distinguish between two topologies provided all node
subspaces are distinct.  This is used to identify the unique parent of each node. 
In general topologies, it is similarly sufficient to 
identify the parents of each node, in order to learn the graph topology.
The following theorem claims that having distinct subspaces is in fact a sufficient condition for topology
identifiability over general graphs as well.

\begin{theorem}\label{th_general_cond}
  In a synchronous network employing randomized network coding over
  $\mathbb{F}_q$, a sufficient condition to uniquely identify the
  topology with high probability as $q\gg 1$, is that
\begin{equation} 
\label{eq_cond_general}
\Pi_u(t) \neq \Pi_v(t) \quad \forall\;  u,v\in V,\quad u\ne v,
\end{equation}
for some time $t$. Under this condition, we can identify the topology by collecting global
information at times $t$ and $t+1$, {\it i.e.}, two consecutive static
views of the network.
\end{theorem}
\begin{proof}
  Assume node $u$ has the $p$ parents
  $P(u)=\{u_1,\ldots,u_p\}$. Let
  $\Pi^{(u_1)}_u(t),\ldots,\Pi^{(u_p)}_u(t)$ denote the subspaces
  node $u$ has received from its parents up to time $t$, where
  $\Pi_u(t)=\sum_{i=1}^{p} \Pi^{(u_i)}_u(t)$. From construction it
  is clear that $\Pi^{(u_i)}_u(t+1)\subseteq\Pi_{u_i}(t)$.

  To identify the network topology, it is sufficient to decide which
  node $v\in V$ is the parent that sent the subspace
  $\Pi^{(u_i)}_{u}(t)$ to node~$u$ for each $i$, and thus find the
  $p$ parents of node~$u$. We claim that, provided
  (\ref{eq_cond_general}) holds, node $u$ has as parent the node $v$
  which at time $t$ has the smallest dimension subspace containing
  $\Pi^{(u_i)}_u(t+1)$. Thus we can uniquely identify the network
  topology, by two static views, at times $t$ and $t+1$, as
  Algorithm~\ref{alg_ind_general} describes.

  Indeed, let $\pi^{(u_i)}_u(t)$ denote the subspace that node $u$ receives 
  from parent $u_i$ at exactly time $t$, that is, 
  \mbox{$\Pi^{(u_i)}_u(t+1)=\Pi^{(u_i)}_u(t) + \pi^{(u_i)}_u(t+1)$}.
  For each $i\in\{1,\ldots,p\}$, if $\pi^{(u_i)}_u(t+1) \nsubseteq \Pi_v(t)$ for all 
  $v\in V \setminus \{u_i\}$, clearly $\Pi^{(u_i)}_u(t+1) \nsubseteq \Pi_v(t)$ 
  for all $v\in V \setminus \{u_i\}$, and we are done.
  Otherwise, using Lemma~\ref{lem:subset} and because
  (\ref{eq_cond_general}) holds, with high probability we have
  $\pi_u^{(u_i)}(t+1) \nsubseteq \Pi_v(t)$ for all $v\in V$ except
  those nodes that their subspaces contain $\Pi_{u_i}(t)$. So we are
  done.
\end{proof}

Note that to identify the network topology, we need to know, for all
nodes $u$, the dimension $d_u\triangleq \dim(\Pi_u(t))$ of their observed subspaces at time $t$, the
dimension $d_u^{(i)}\triangleq \dim(\Pi^{(u_i)}_u(t+1))$ for
all parents~$u_i$ of node~$u$, and the dimension of the intersection of
$\Pi^{(u_i)}_u(t+1)$ with all $\Pi_w(t)$, $w\in V$, denoted as 
\mbox{$d_{w u}^{(i)}\triangleq \dim(\Pi^{(u_i)}_u(t+1)\cap \Pi_w(t))$}. 
Algorithm~\ref{alg_ind_general} uses this information to infer the topology.
\begin{center}
  \begin{pseudocode}[doublebox]{Gen}{\{d_u\},\{ d_u^{(i)} \}, \{ d_{w u}^{(i)} \}}
\FOREACH u \in V \DO 
\BEGIN
\IF d_u=n \THEN  u\GETS S
\ELSE
\BEGIN
\FOREACH i \in \{1,\ldots,p_u\} \DO
\BEGIN
\text{node $u$ has as parent the}\\
\text{ node $v$ with}\\
v= \underset{ w\in V:\ d_{w u}^{(i)} = d_u^{(i)} }{\operatorname{arg\,min}} d_w
\END
\END
\END
\label{alg_ind_general}
\end{pseudocode}
\\
\vspace{-0.3cm}
Alg.~\ref{alg_ind_general}: Find the topology of a general network.
\end{center}

The sufficient conditions (\ref{eq_cond_general}) in
Theorem~\ref{th_general_cond}, may or may not hold, depending on the
network topology and the information dissemination protocol. Next, we
will investigate for what network topologies the conditions (\ref{eq_cond_general}) 
hold for the dissemination Algorithm~\ref{alg_diss} so that the network is
identifiable.

\begin{lemma}\label{lem:equalcond}
  Consider two arbitrary nodes $u$ and $v$, where
  $P(u)=\{u_1,\ldots,u_{p_u}\}$ and $P(v)=\{v_1,\ldots,v_{p_v}\}$ are
  the parents of $u$ and $v$ respectively. Let
  $\Pi_{P(u)}(t-1)=\sum_{i=1}^{p_u}\Pi_{u_i}(t-1),$ and 
  $\Pi_{P(v)}(t-1)=\sum_{i=1}^{p_v}\Pi_{v_i}(t-1).$ 
  If $\Pi_u(t)=\Pi_v(t)$ we should have had
  $\Pi_{P(u)}(t-1)=\Pi_{P(v)}(t-1)$ w.h.p.
\end{lemma}
\begin{proof}
  Suppose $\Pi_{P(u)}(t-1)\neq\Pi_{P(v)}(t-1)$ and let us assume that
  $\Pi_u(t)=\Pi_v(t)=\Pi$. This implies that if $\pi_u(t)$ and
  $\pi_v(t)$ are subspaces collected by nodes $u$ and $v$ at time $t$
  then,
\begin{align*}
\Pi_u(t) &= \Pi_v(t) = \Pi \\
\pi_u(t) + \Pi_u(t-1) &= \pi_v(t) +\Pi_v(t-1).
\end{align*}
From construction, we have $\Pi=\Pi_u(t)\subseteq\Pi_{P(u)}(t-1)$ and 
$\Pi=\Pi_v(t)\subseteq\Pi_{P(v)}(t-1)$.

On the other hand, since we randomly chose $\pi_u^{(u_i)}(t)$ from
$\Pi_{u_i}(t-1)$ and since $\pi_u^{(u_i)}(t)\subseteq\Pi$ 
(because $\pi_u(t)\subseteq\Pi$) using Lemma~\ref{lem:subset}
we conclude that we should have that $\Pi_{u_i}(t-1)\subseteq\Pi$
which means we should have $\Pi_{P(u)}(t-1)\subseteq\Pi$. Similarly,
we should have $\Pi_{P(v)}(t-1)\subseteq\Pi$. As a result (w.h.p.) we have to have 
\begin{equation*}
\Pi_{P(u)}(t-1)=\Pi_{P(v)}(t-1)=\Pi,
\end{equation*}
which is a contradiction, so we are done.
\end{proof}

\begin{corollary}\label{cor:equalcond_gener}
  If $\Pi_u(t)=\Pi_v(t)=\Pi$ for $t>l$ we should have had
  $\Pi_{P^l(u)}(t-l)=\Pi_{P^l(v)}(t-l)=\Pi$, w.h.p.
\end{corollary}

\begin{proof}
  Consider the parents of nodes $u$ and $v$ as supernodes $P(u)$ and
  $P(v)$. Using a similar argument as stated in
  Lemma~\ref{lem:equalcond}, we can conclude that the parents of
  $P(u)$ and $P(v)$, denoted as $P^2(u)$ and $P^2(v)$, should satisfy
\begin{equation*}
\Pi_{P^2(u)}(t-2)=\Pi_{P^2(v)}(t-2)=\Pi.
\end{equation*}
We use this argument $l$ times to get the result.
\end{proof}

\begin{lemma}
  If the dissemination protocol is in the steady state, $t\ge T_s$,
  we could not have $\Pi_u(t)=\Pi_v(t)$ unless nodes $u$ and $v$ have
  the same set of ancestors at some $l$ level above in the network.
\end{lemma}

\begin{proof}
  Because $t\ge T_s$, we have $d_u=\dim(\Pi_u)<n$ and
  $d_v=\dim(\Pi_v)<n$. Let us assume $\Pi_u(t)=\Pi_v(t)=\Pi$ so we
  have $d\triangleq d_u=d_v$. From the Corollary~\ref{cor:equalcond_gener} we
  can write
\[
\Pi_{P^l(u)}(t-l)=\Pi_{P^l(v)}(t-l)=\Pi,
\]
for every $l\ge 1$. Increasing $l$, two cases may happen. First,
either $P^l(u)$ or $P^l(v)$ contains the source node $S$ that results
in $\dim(\Pi_{P^l(u)}(t-l))=n$ or $\dim(\Pi_{P^l(v)}(t-l))=n$ which is
a contradiction since $d<n$. Second, nodes $u$ and $v$ have the same
set of ancestors at some level $l$.
\end{proof}

Up to here, we have shown that assuming the dissemination protocol is
in the steady state the subspaces of two arbitrary nodes are equal
only if they have the same ancestors at some level above in the
network. 
The following result, Theorem~\ref{thm:SuffCond} states
sufficient conditions that make the nodes' subspace different for
dissemination Algorithm~\ref{alg_diss}.
\begin{theorem}\label{thm:SuffCond}
  Suppose two arbitrary nodes $u$ and $v$ have the same set of parents
  $P^l=P^l(u)=P^l(v)$ at some level $l$. The following conditions are
  sufficient so that the dissemination Algorithm~\ref{alg_diss} satisfies
  condition (\ref{eq_cond_general})\footnote{Note that  the min-cut to node $u$, $c_u=\textrm{min-cut}(S,u)$, equals
    $c_u=\min\{\hat{c}_u,c_p\}$.}:
\begin{align*}
\hat{c}_u=\textrm{min-cut}(P^l,u) &\le \textrm{min-cut}(S,P^l)=c_p,\\
\hat{c}_v=\textrm{min-cut}(P^l,v) &\le \textrm{min-cut}(S,P^l)=c_p.
\end{align*}
\end{theorem}

\begin{proof}
Consider the set of nodes in $P^l$. From the definition we know that
there exists at least one path of length $l$ from each node in $P^l$ to 
the node $u$. But also there might exist paths of length less than $l$
from some nodes in $P^l$ to $u$. If this is the case, because the topology 
is a directed acyclic graph, we can find a subset $P'$
of the nodes in $P^l$ such that it forms a cut for the node $u$ and the shortest
path from each node in $P'$ to $u$ is $l$; see Figure~\ref{fig:GenTopThmProofFigure}. Moreover, we have 
$\textrm{min-cut}(S,P')=c_p$ and $\textrm{min-cut}(P',u)=\hat{c}_u$.

Now assume that $P'=\{p_1,\ldots,p_k\}$ such that $\tau_{p_1}\le\cdots\le\tau_{p_k}$.
Let $a_1,\ldots,a_k,$ be the accumulative min-cut from $S$ to each node in $P'$.
By this we mean that $a_1=c_{p_1}$ and $a_2$ is the amount of increase in the 
min-cut from $S$ by adding node $p_2$ and so on. We similarly consider  
the accumulative min-cut values from $p_i$ to $u$ and denote these by $b_1,\ldots,b_k$. So we have
$\sum_{j=1}^k a_j=c_p$ and $\sum_{j=1}^k b_j=\hat{c}_u$.

From definition of the waiting times (Definition~\ref{def:waiting_time}) 
we can write
\begin{gather*}
d_{P'}(\tau_1) \ge a_1+1,\\
d_{P'}(\tau_2) \ge d_{P'}(\tau_1) +(\tau_2-\tau_1)a_1 + a_2,\\
d_{P'}(\tau_k) \ge d_{P'}(\tau_{k-1}) +(\tau_{k}-\tau_{k-1})\sum_{j=1}^{k-1} a_j + a_k.
\end{gather*}
Then we have
\begin{gather}
d_{P^l}(\tau_k) \ge d_{P'}(\tau_k)\nonumber\\
\ge (\tau_2-\tau_1)a_1+\cdots+(\tau_k-\tau_{k-1})\sum_{j=1}^{k-1} a_j + \sum_{j=1}^k a_j + 1. \label{eq:AccMinCut-dP}
\end{gather}

For $d_u$ we can also write
\begin{gather*}
d_u(\tau_1+l) \le b_1,\\
d_u(\tau_2+l) \le d_u(\tau_1+l) + (\tau_2-\tau_1)\min[a_1,b_1] + b_2,\\
d_u(\tau_k+l) \le d_u(\tau_{k-1}) + (\tau_k-\tau_{k-1})\min[\sum_{j=1}^{k-1} a_j,\sum_{j=1}^{k-1}b_j] + b_k,
\end{gather*}
or
\begin{gather}
d_u(\tau_k+l) \le (\tau_2-\tau_2)\min[a_1,b_1] \nonumber\\
 +\cdots+ (\tau_k-\tau_{k-1})\min[\sum_{j=1}^{k-1} a_j,\sum_{j=1}^{k-1}b_j] + \sum_{j=1}^k b_j.\label{eq:AccMinCut-du}
\end{gather}

From \eqref{eq:AccMinCut-dP}, \eqref{eq:AccMinCut-du} and the theorem assumptions we 
conclude that $d_u(\tau_k+l) < d_{P^l}(\tau_k)$. Now for $\Delta t$ timeslots later
we write
\begin{align*}
d_u(\tau_k+l+\Delta t) &\stackrel{\text{(a)}}{\le} d_u(\tau_k+l) + \hat{c}_u \Delta t\\
&\stackrel{\text{(b)}}{<} d_{P^l}(\tau_k) + c_p \Delta t\\
&\stackrel{\text{(c)}}{=} d_{P^l}(\tau_k+\Delta t),
\end{align*}
where (a) is true because $u$ receives packets from $P^l$ with rate at most $\hat{c}_u$;
(b) is true because $d_u(\tau_k+l) < d_{P^l}(\tau_k)$ and $\hat{c}_u\le c_p$; and finally
(c) is true because after $\tau_k$ all of the nodes in $P'$ receive packets at rate 
equal to their min-cut which means that $P'$ (the same is true for $P^l$) receives packets
at rate equal to its min-cut $c_p$.

The same inequality holds for the dimension of
\mbox{$\Pi_v(\tau_k+l+\Delta t)$}. Thus for time $t>\tau_k+l$ we cannot have
\mbox{$\Pi_{P^l}(t-l)=\Pi_u(t)$} and
\mbox{$\Pi_{P^l}(t-l)=\Pi_v(t)$} if $\hat{c}_u\le c_p$ and $\hat{c}_v\le c_p$. So using
Corollary~\ref{cor:equalcond_gener} we are done.
\end{proof}

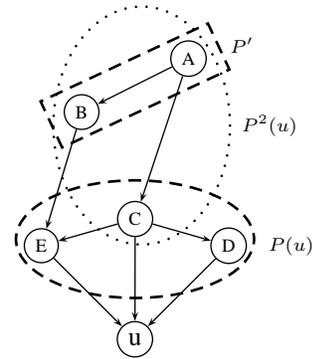
\begin{figure}[hbt]
\begin{center}
\psset{unit=0.07in}
\begin{pspicture}(0,-1)(20,25)
\psset{linewidth=0.5pt}


\rput(17,6){\circlenode{D}{\scriptsize{D}}}
\rput(3,6){\circlenode{E}{\scriptsize{E}}}
\rput(10,8){\circlenode{C}{\scriptsize{C}}}
\rput(10,-1){\circlenode{u}{u}}
\ncline{->}{D}{u}
\ncline{->}{E}{u} 
\ncline{->}{C}{u} 
\ncline{->}{C}{D} 
\ncline{->}{C}{E}  

\psellipse[linewidth=1pt,linestyle=dashed,hatchangle=0](10,6.5)(9,4.5)
\rput[Cl](20,6){\scriptsize{$P(u)$}}

\rput(14,20){\circlenode{A}{\scriptsize{A}}}
\rput(6,16){\circlenode{B}{\scriptsize{B}}}
\ncline{->}{A}{B}
\ncline{->}{B}{E}
\ncline{->}{A}{C}

\psellipse[linewidth=1pt,linestyle=dotted,hatchangle=0](10.5,15)(6.7,9)
\rput[Cl](18,15){\scriptsize{$P^2(u)$}}

\rput{25}(10,18){\psframe[linewidth=1pt,linestyle=dashed](-7,2)(7,-2)}
\rput(18,21){\scriptsize{$P'$}}
\end{pspicture}
\end{center}
\caption{Sets used in the proof of Theorem~\ref{thm:SuffCond}: the set $P(u)$ contains the parents of node $u$ at distance $l=1$; the set $P^2(u)$ contains the set of parents at distance  $l=2$; while $P'$ is the subset of  $P^2(u)$ at distance no less than $l=2$.}

\label{fig:GenTopThmProofFigure}
\end{figure}

Intuitively, what Theorem~\ref{thm:SuffCond} tell us is that, if for a
node $u$ there exists a path that does not belong in any cut between
the source and another node $v$, then nodes $u$ and $v$ will
definitely have distinct subspaces. The only case where nodes $u$ and
$v$ may have the same subspace is, if they have a common set of
parents, a common cut. Even then, they would need both of them to
receive all the innovative information that flows through the common
cut at the same time. Note that the condition of
Theorem~\ref{thm:SuffCond} are also necessary for identifiably for
the special case of tree topologies, such as the topology in
Figure~\ref{fig_tree}.

\subsection{Practical Considerations}\label{sec_practical}
We here argue that our proposed scheme can  lead to a practical protocol,
where nodes passively collect information during the dissemination, and send once a small amount of information
to the central node in charge of the topology inference.
In particular, we assume that the
 nodes follow the information dissemination protocol and at some point
the central node query them to report the subspaces they  gather at a 
specific\footnote{We assume the query is send before time $t$ actually occurs; 
Also note that if the number of source packets $n$ is much larger than the 
min-cut to each node, and if we have an estimate for $\Delta_{\text{i}}(G)$, 
a central node can with high probability select at time $t$ in steady state. 
A node can also send a feedback message to inform the central node if it is 
not at steady state at time $t$.} time $t$. 

We now calculate the communication cost (total number of 
bits required to be transmitted to a central node) of the proposed passive
inference algorithm. 
Each node has  to transmit at most $2\Delta_{\text{i}}(G)$ subspaces to 
the central node where $\Delta_{\text{i}}(G)$ is the maximum in-degree of nodes in the 
network. There are $\vartheta$ nodes in the network so $2\vartheta \Delta_{\text{i}}(G)$ 
subspace have to be transmitted. 
The total number of subspaces of $\Pi_S$ (which itself is an $n$-dimensional
space) is
\[
\sum_{i=1}^{n} \gaussnum{n}{i}{q} \approx \sum_{i=1}^n q^{i(n-i)} \approx q^{n^2/4},
\]
where $\gaussnum{n}{i}{q}$ is the Gaussian number, the number of $i$-dimensional
subspaces of an $n$-dimensional space. To approximate the Gaussian number we use 
\cite[Lemma~1]{JaMoFrDi-IT11}; note that the approximation holds for large $q$.

So to encode one of the subspace of $\Pi_S$ we need approximately $\frac{n^2}{4} \log_2{q}$ bits.
As a result, the total number of bits need to be transmitted to the central node
is at most 
\[
\frac{2n^2 \Delta_{\text{i}}(G) \vartheta}{4} \log_2{q}.
\] 

Clearly, the complexity depends on the size of $n$, 
the number of packets that the source transmits. 
In our work we assume that $n$ is large enough,
so that the network enters in steady state; on the other hand, other 
considerations such as decoding complexity at network nodes, 
would require $n$ to take moderate values.
Note that, for our algorithm to work, 
(\ie, to sample the network while in the steady state)
we only require that $n=2\beta c_{\text{max}} D(G)$ (Corollary~\ref{cor:SuffCondon_n}),
where $\beta>1$ is some constant that
determines how many time slots the network is in the steady state.
If $n$ has such a size, the maximum number of bits that
need to be transmitted per node (communication cost per node) is
\[
R_{\text{com-cost/ND}} \approx 2\beta^2 c_{\text{max}}^2 D(G)^2 \Delta_{\text{i}}(G) \log_2{q} \quad\text{bits}.
\]
In the above equation $\beta$, $c_{\text{max}}$, and $\Delta_{\text{i}}(G)$
are some constants. The only parameter that depends on the network 
size is $D(G)$. However for the most of practical
content distribution networks the longest path of network is kept small 
to ensure a good connectivity between nodes in the network (see for 
example \cite{HaLeBa-INRIA}).

To give a specific example for a possible communication cost, let us consider
a practical scenario where $q=2^8$, $c_{\text{max}}=1$, $\beta^2=5$, 
$\Delta_{\text{i}}(G)=5$, and $D(G)=10$. Then we have $R_{\text{com-cost/nd}}\approx 4$
kilobytes. In contrast, in a practical dissemination scenario (ex. of video) we would disseminate a large number of information packets each  possibly
as large as a few megabytes; thus the overhead of the topological information would not be significant.

\section{Locating Byzantine Attackers}\label{chp:ByzantineAttack}

In this section we explore a problem that is dual to topology
inference: given complete knowledge of the topology, we leverage
subspace properties to identify the location of a malicious Byzantine
attacker.

In a network coded system, the adversarial nodes in the network
disrupt the normal operation of the information flow by inserting
erroneous packets into the network. This can be done by inserting
spurious data packets into their outgoing edges. One way in which
these erroneous packets can be prevented from disrupting information
flow is by reducing the transmission rate to below the min-cut of the
network, and using the redundancy to protect against errors; 
\cite{YeCa2006_1,CaYe2006_2,Zh2006}. One such
technique, using subspaces to code information was proposed in
\cite{KoKs2007}. In this approach, the source sends a basis of the
subspace corresponding to the message. In the absence of errors, the
linear operations of the intermediate nodes do not alter the sent
subspace, and hence the receiver decodes the message by collecting the
basis of the transmitted subspace. A malicious attacker inserts
vectors that do not belong in the transmitted subspace. Therefore, if
the message codebook uses subspaces that are ``far enough'' apart
(according to an appropriately defined distance measure), then one can
correct these errors \cite{KoKs2007}. Note that in this technique, we
do not need any knowledge of the network topology for the error
correction mechanism. All that is needed is that the intermediate
nodes do not alter the transmitted subspace (which can be done if they
do linear operations).

The approach of this section to locating adversaries uses the
framework developed in the previous sections, where it was shown that
under randomized network coding, the subspaces gathered by the nodes
of the network provide information about the topology. Therefore, the
basic premise in this section is to use the structure of the erroneous
subspace inserted by the adversary to reveal information about its
location, when we already know the network topology.

\subsection{Problem Formulation}\label{sec:Byz_Attack:ProbForm}

Consider a network represented as a directed acyclic graph
$G=(V,E)$. We have a source, sending information to $r$ receivers, and
one (or more) Byzantine adversaries, located at intermediate nodes of
the network. We assume complete knowledge of the network topology, and
consider the source and the receivers to be trustworthy
(authenticated) nodes, that are guaranteed not to be adversaries.

Suppose source $S$ sends $n$ vectors, that span an $n$-dimensional
subspace $\Pi_S$ of the space $\Fbb_q^\ell$, where we assume $q\gg 1$. In particular, 
in this section we will consider (without loss of generality) subspace coding, where
$\Pi_S$ belongs to a codebook $\mathcal{C}$, $\Pi_S \in \mathcal{C}$
designed to correct network errors  and erasures \cite{KoKs2007}.

In the absence of any adversaries in the network each
receiver $R_i$, $i=1,\ldots,r$, can decode the exact space $\Pi_S$. Now assume that 
there
is an adversary, Eve, who attacks one of the nodes in the network by
combining a $\delta$-dimensional subspace $\Pi_\varepsilon$ with its
incoming space and sending the resulting vectors to its children.
Then the receiver $R_i$ collects some linearly independent vectors that span a subspace
$\Pi_{R_i}$. We can write
\[ 
\Pi_{R_i} = \mathcal{H}_i(\Pi_S + \Pi_\varepsilon), 
\]
where $\mathcal{H}_i(\Pi)$ is a linear operator. 
This operator models the linear transformation that the network induces
on the inserted source and adversary packets.

We assume that the receiver is able to at least detect that a
Byzantine attack is under way. Moreover, we assume that the receiver
is able to decode the subspace $\Pi_S$ that the source has sent. This
might be, either because the receiver has correctly decoded the sent
message (\ie, using code construction from \cite{KoKs2007}), or, 
because after detecting the presence of an attack has
requested the source subspace through a secure channel from the source
node.

We can restrict the Byzantine attack in several ways, depending on the
edges where the attack is launched, the number of corrupted vectors
inserted, and the vertices (network nodes) that the adversary has
access to. In this section we will distinguish between the cases where
\begin{enumerate}
\item[I.] there is a single Byzantine attacker located in a vertex of the network, and  
\item[II.] there are multiple independent attackers, located on
  different vertices, that act without coordinating with each other.
\end{enumerate}
We assume that each attacker located on a single vertex is able to
corrupt any outgoing edges by inserting arbitrary erroneous information.
However, in this work we only consider the case where the attackers
inject independent information without any coordination among themselves.


We are interested in understanding under what conditions we can
uniquely identify the attacker's location (or, up to what uncertainty
we can identify the attacker), under the above scenarios.

\subsection{The Case of a Single Adversary}\label{sec_single} 
In this section we focus on the case where we want to locate a
Byzantine adversary, Eve, controlling a {\em single} vertex of the network
graph.  

In \S\ref{subsec:OnlyTop} we illustrate the limitation of using {\em
  only} the information the receivers have observed along with the
knowledge of the topology, to locate the adversary. This motivates
requiring additional information from the intermediate nodes related
to the subspaces observed by them. In \S\ref{subsec:InfoIntNodes}, we
show that such additional information allows us to localize the
adversary either uniquely or within an ambiguity of at most two nodes.

\subsubsection{Identification using only Topological Information}\label{subsec:OnlyTop}
In order to illustrate the ideas, we will examine the case where
the corrupted packets are inserted on a single edge of the network,
say edge $e_A$. The extension to the cases where multiple edges get
corrupted is easy.

Since each receiver $R$ knows the subspaces $\{\Pi^{(i)}_R\}$ it has
received from its $|\In(R)|$ parents, it knows whether what it
received is corrupted or not (a subspace of $\Pi_S$ or not). Using
this, we can infer some information regarding topological properties
that the edge $e_A$ should satisfy. In particular we have the following
result, Lemma~\ref{lem:IdenUsingTopInfo}.

\begin{lemma}\label{lem:IdenUsingTopInfo}
Let $P_e$ denote the set of paths\footnote{In the following we are going
to equivalently think of $P_e$ as the set of all edges that take
part in these paths.} starting from the source and ending at edge
$e$.
Then, if $\mathcal{E}_C$ is the set of incoming edges to receivers
that bring corrupted packets, while $\mathcal{E}_S$ the set of
incoming edges to receivers that only bring source information, the
edge $e_A$ belongs in the set of edges $\mathcal{E}_A$, with
\[  
\mathcal{E}_A \triangleq \left\{ \bigcap_{e \in \mathcal{E}_C} P_e -
\bigcup_{e \in \mathcal{E}_S} P_e \right\}.
\] 
\end{lemma}

\begin{proof}
If $R$ receives corrupted vectors from an incoming edge $e$ then
there exists at least one path that connects $e_A$ to $e$. 
Then $e_A$ is part of at least one path in $P_e$.

Conversely, if a receiver $R$ does not receive corrupted packets
from an incoming edge $e$, then $e_A$ does not form part of any path
in $P_e$. That is, there does not exist a path that connects $e_A$
to $e$.
\end{proof}

The following example illustrates this approach.
\begin{example}\label{ex:ByzAttack:ex1}
  Consider the network in Figure~\ref{fig_ex1}, and assume that $R_1$
  receives corrupted packets from edge $DR_1$ and uncorrupted packets
  from $AR_1$, while $R_2$ receives only uncorrupted packets.
\begin{figure}[h]
\vspace{-1.5em}
\begin{center}
\psset{unit=0.035in}
\begin{pspicture}(-15,-5)(35,40)
\psset{linewidth=0.5mm}
\begin{small}
\rput(10,30){\circlenode{S}{S}}
\rput(-5,8){\circlenode{A}{A}}
\rput(10,15){\circlenode{B}{B}}
\rput(25,8){\circlenode{C}{C}}
\rput(10,-5){\circlenode{D}{D}}
\rput(-15,-5){\circlenode{R1}{$R_1$}}
\rput(35,-5){\circlenode{R2}{$R_2$}}
\ncline[linewidth=0.5mm,linecolor=black]{->}{S}{A}
\ncline[linewidth=0.5mm,linecolor=black]{->}{S}{B}
\ncline[linewidth=0.5mm,linecolor=black]{->}{S}{C}
\ncline[linewidth=0.5mm,linecolor=black]{<-}{A}{B}
\ncline[linewidth=0.5mm,linecolor=black]{->}{A}{D}
\ncline[linewidth=0.5mm,linecolor=black]{->}{B}{C}
\ncline[linewidth=0.5mm,linecolor=black]{->}{C}{D}
\ncline[linewidth=0.5mm,linecolor=black]{->}{A}{R1}
\ncline[linewidth=0.5mm,linecolor=black]{->}{D}{R1}
\ncline[linewidth=0.5mm,linecolor=black]{->}{C}{R2}
\ncline[linewidth=0.5mm,linecolor=black]{->}{D}{R2}
%
%
%
\end{small}
\end{pspicture}
\end{center} 
\caption{The source $S$ distributes packets to receivers $R_1$ and $R_2$.}
\label{fig_ex1}
\end{figure}
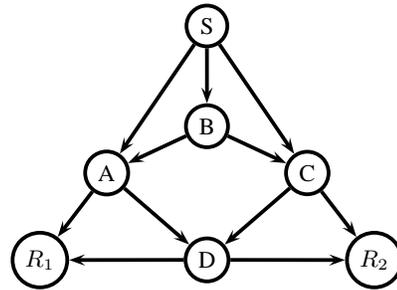
Then $\mathcal{E}_A=\{DR_1\} $ and the attacker is located on node $D$. 
\hfill{$\blacksquare$}
\end{example}

In Example~\ref{ex:ByzAttack:ex1}, we were able to exactly identify
the location of the adversary, because the set $\mathcal{E}_A$
contained a single edge, and node $R_1$ is trustworthy. It is easy to
find network configurations where $\mathcal{E}_A$ contains multiple
edges, or in fact all the network edges, and thus we can no longer
identify the attacker. The following example illustrates one such
case.  

\begin{example}\label{ex:ByzAttack:ex2}
  Consider the line network shown in
  Figure~\ref{line_network}. Suppose the attacker is node $A$. If the
  receiver $R$ sees a corrupted packet, then using just the topology,
  the attacker could be {\em any} of the other nodes in the line
  network. This illustrates that just the topology and receiver
  information could lead to large ambiguity in the location of the
  attacker.  \hfill{$\blacksquare$}
\end{example}

Therefore, Example~\ref{ex:ByzAttack:ex2} motivates the ideas examined
in \S\ref{subsec:InfoIntNodes} which obtain additional information and
utilize the structural properties of the subspaces observed.

\subsubsection{Identification using Information from all Network Nodes}
\label{subsec:InfoIntNodes}
We will next discuss algorithms where a central authority, which we
will call {\em controller}, requests from all nodes in the network to
report some additional information, related to the subspaces they have
received from their parents.  The adversary could send inaccurate
information to the controller, but the other nodes report the
information accurately. Our task is to design the question to the
nodes such that we can locate the adversary, despite its possible
misdirection.

The controller may ask the nodes of the following types of
information, listed in decreasing order of complexity:
\begin{itemize}
\item[]{\em Information 1:} Each node $v$ sends all subspaces
  $\Pi^{(i)}_v$ it has received from its parents, where
  $\Pi_v=\sum_{i\in P(v)}\Pi^{(i)}_v$.
\item[]{\em Information 2:} Each node $v$ sends a randomly chosen
  vector from each of the received subspaces $\Pi^{(i)}_v$
  ($|\In(v)|$ vectors in total).
\end{itemize}

Information 2 is motivated by the following well-known
observation, see Lemma~\ref{lem:subset}: let $\Pi_1$ and $\Pi_2$ be
two subspaces of $\mathbb{F}_q^n$, and assume that we randomly select
a vector $\Bs{y}$ from $\Pi_1$.  Then, for $q\gg 1$, $\Bs{y}\in \Pi_2$ if and
only if $\Pi_1\subseteq \Pi_2$. Thus, a randomly selected vector from
${\Pi}_v$ allows to check whether ${\Pi}_v\subseteq
\Pi_S$ or not.

In fact, we will show in this section that for a single adversary it
is sufficient to use\footnote{Using Information~2 these
statements are made with high probability, \ie, the
probability goes to one as field size $q\rightarrow\infty$.}
Information~2, and classify the edges of the network by simply testing
whether the information flowing through each edge is a subspace of
$\Pi_S$ or not (\ie, is corrupted or not).

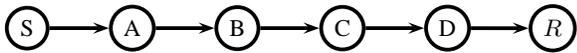
\begin{figure}[h]
\vspace{-0.8em}
\begin{center}
\psset{unit=0.055in}
\begin{pspicture}(-5,-5)(55,5)
\psset{linewidth=0.5mm}
\begin{small}
\rput(0,0){\circlenode{S}{S}}
\rput(10,0){\circlenode{A}{A}}
\rput(20,0){\circlenode{B}{B}}
\rput(30,0){\circlenode{C}{C}}
\rput(40,0){\circlenode{D}{D}}
\rput(50,0){\circlenode{R}{$R$}}
\ncline[linewidth=0.5mm,linecolor=black]{->}{S}{A}
\ncline[linewidth=0.5mm,linecolor=black]{->}{A}{B}
\ncline[linewidth=0.5mm,linecolor=black]{->}{B}{C}
\ncline[linewidth=0.5mm,linecolor=black]{->}{C}{D}
\ncline[linewidth=0.5mm,linecolor=black]{->}{D}{R}
\end{small}
\end{pspicture}
\end{center} 
\vspace{-1.5em}
\caption{ 
The source $S$ sends information to receiver $R$ over a line network.
}
\label{line_network}
\end{figure}

\begin{theorem}\label{thm:FindAdvBySplittingSnglCase}
Using Information~1, by splitting the network edges into corrupted and
uncorrupted sets, we can narrow the location of the adversary up
to a set of at most two nodes. With Information 2, the same result
holds w.h.p.
\end{theorem}
\begin{proof}
The network is a directed acyclic graph, so we can impose a partial order on
the edges of the graph, such that $e_1>e_2$ if $e_1$ is an ancestor
edge of $e_2$ (\ie, there exists a path from $e_1$ to $e_2$).
Then having Information~1 or Information~2, we can divide the edges of 
the network into two sets: the set of edges $E_C$ through which are 
reported to flow corrupted subspaces, and the remaining edges $E_S$ 
through which the source information flows so we have $E=E_S\cup E_C$ 
and $E_S\cap E_C=\emptyset$. Note that all the outgoing 
edges from the source belong in $E_S$.

Nodes in the network perform randomized network coding so every
node that receives corrupted information on at least one of its incoming
edges makes all of the outgoing edges polluted w.h.p. Let $t_v$ be the
number of corrupted outgoing edges of a node $v$ where we have
$1\le t_v\le|\Out(v)|$. For each node $v$ that is not an adversary
we have either $t_v=0$ or $t_v=|\Out(v)|$.

Now, to prove the theorem we consider the following possible cases.
\begin{enumerate}
\item If the adversary Eve corrupts $t_A$ outgoing edges where $1<t_A<|\Out(A)|$
we can identify  the node she has attacked uniquely because its  behavior is different from
all other nodes.
\item If she  corrupts all of its outgoing edges, $t_A=|\Out(A)|$, 
then she can fraud us by declaring that one of the node's 
incoming edges is corrupted. If $A$ declares more than one of
the incoming edges as corrupted we can find its location uniquely.
\item She can also corrupt only one of its outgoing edges, $t_A=1$, and pretends
that its children is in fact the adversary by declaring all of 
its incoming edges bring non-corrupted information. She cannot
declare that any of its incoming edges are polluted since then 
we may find its location uniquely.
\end{enumerate}

In all of the above cases the adversary is on the boundary of two
sets $E_S$ and $E_C$ and the ambiguity about its location is at 
most withing a set of two vertices where this set contains those
two vertices that are connected by the corrupted edge with highest 
order among all corrupted edges (recall that we can compare all of the 
corrupted edges using the imposed partial order).
\end{proof}

\subsection{The Case of Multiple Adversaries}\label{sec_multiple}
In the case of a single adversary, it was sufficient to divide the set
of edges into two sets, $E_S$ and $E_C$, as described in the previous
section. In the presence of multiple adversaries, this may no longer
be sufficient. An additional dimension is that realistically, we may
not know the exact number of adversaries present. In the following, we
discuss a number of algorithms, that offer weaker or stronger
identifiability guarantees.

\subsubsection{Identification using only Topological Information}
The approach in \S\ref{subsec:OnlyTop} can be directly extended in the
case of multiple adversaries, but again, offers no identifiability
guarantees.
\begin{example}
Consider again the network in Figure~\ref{fig_ex1}, and assume 
that $R_1$ receives corrupted packets only from  edge $DR_1$ 
while $R_2$ receives corrupted packets only from edge  $DR_2$. 
Then \mbox{$\mathcal{E}_A=\{AD,CD,DR_1,DR_2\}$} and (depending 
on our assumptions) we may have,
\begin{itemize}
\item[-]a single adversary located on node $D$,
\item[-] two adversaries, located on nodes $A$ and $C$,
\item[-] two adversaries, located on nodes $A$ and $D$, or nodes $C$ and $D$, or
\item[-] three adversaries, located on nodes $A$, $C$, and $D$.
\end{itemize}
\hfill{$\blacksquare$}
\end{example}

\subsubsection{Identification using Splitting}
Similarly to \S\ref{subsec:InfoIntNodes}, using Information~1 or 
Information~2, we can divide the set of
edges into two sets $E_S$ and $E_C$, depending on whether the
information flowing through each edge belongs in $\Pi_S$ or
not. Depending on the network topology, we may be able to uniquely
identify the location of the attackers. However, this approach,
although it guarantees to find at least one of the attackers (within
an uncertainty of at most two nodes), does not necessarily find all
the attackers, even if we know their exact number.

To show this let us state the following definition.
\begin{definition}
We say that node $v$ is in the shadow of node $A$, if there 
exists a path that connects every incoming edge of $v$ to a 
corrupted outgoing edge of $A$. 
\end{definition}
Then we have the following result.
\begin{lemma}
By splitting the network edges into two sets $E_S$ and $E_C$
we cannot   identify  adversarial nodes that are in the
shadow of an adversary $A$.
\end{lemma}
\begin{proof}
This is because if an attacker is in the shadow of
another attacker, it may corrupt only already corrupted
vectors and thus not incur a detectable effect.
So we cannot distinguish between an attacker
and a normal node that are in the shadow of $A$.
\end{proof}

The following example illustrates these points.
\begin{example}
For the example in Figure~\ref{fig_ex1}, assume that each attacker 
corrupts all its outgoing edges, and consider the following two situations:
\begin{enumerate}
\item Assume that nodes $A$ and $C$ are attackers. If $A$ reports
truthfully while $C$ lies we get $E_C=\{ AD,AR_1, DR_1, DR_2, BC,
CR_2,CD \}$, which allows to identify the attackers.
\item Assume that nodes $B$ and $D$ are attackers. Then we say 
that node $D$ is in the shadow of node $B$, as it corrupts only 
already packets corrupted by $B$. Indeed, if 
$E_C=\{SB, BA,BC, AD,AR_1, DR_1, DR_2, BC, CR_2,CD\}$, knowing 
that the source is trustworthy, we can infer that node $B$ is an 
attacker. However, any of the nodes $A$, $C$, and $D$ can equally 
probably be the second attacker. All these nodes are in the shadow 
of node $D$.
\end{enumerate}
\hfill{$\blacksquare$}
\end{example}

\begin{theorem}\label{thm:FindAdvBySplittingGenCase}
Using Information~1 it is possible to narrow down the location
of those adversaries that have the highest order in the network
using the splitting method. The same result holds for Information~2
w.h.p.
\end{theorem}
\begin{proof}
As stated in the proof of Theorem~\ref{thm:FindAdvBySplittingSnglCase}
we can impose a partial order on the edges of the network graph. Then,
by using Information~1 or Information~2 we may split the network
edges into two sets $E_S$ and $E_C$.

Because every node in the network performs randomized network 
coding, there are only two possibilities for each adversary to 
corrupt its outgoing edges and report subspaces for its incoming 
edges such that it is not located uniquely. These are as follows.
\begin{enumerate}
\item She corrupts some (or all) of its outgoing edges but reports
its incoming edges as uncorrupted.
\item She corrupts all of its outgoing edges and reports some (at 
least one) of its incoming edges as corrupted.
\end{enumerate}

Now, let us consider the set of all the corrupted edges that
have highest order with respect to other corrupted edges and
cannot be compared against each other. For each of the above
cases there should be at least one adversary connected to 
every edge in this set.
\end{proof}

\subsubsection{Identification using Subset Relationships}
In this subsection we develop a new algorithm to find the
adversaries which is based on Information~1.

For each node $u\in V$, let $P(u)=\{u_1,\ldots,u_{p_u} \}$ denote 
the set of parent nodes of $u$. We are going to treat $P(u)$ as 
a super node, and use the notation 
$\Pi_{P(u)}=\sum_{i=1}^{p_u} \Pi_{u_i}$ for the union of the 
subspaces of all nodes in $P(u)$. Also recall that $\Pi_v^{(u)}$ 
denotes the subspace received by node $v$ from node $u$.

Our last algorithm checks, for every node $u\in V$, whether
\begin{equation}\label{eq:SubsetRelation}
\Pi_v^{(u)} \stackrel{?}{\subseteq} \Pi_{P(u)} \quad \forall v\in V: e_{uv} \in E.  
\end{equation}
Then we have the following result, Theorem~\ref{thm:ByzIdntfUsingSubsetRelation}.
\begin{theorem}\label{thm:ByzIdntfUsingSubsetRelation}
If the pairwise distance between adversaries is greater than 
two, it is possible to find the exact number as well as the 
location of the attackers (within an uncertainty of 
parent-children sets) using the subset method.
\end{theorem}

\begin{proof}
First, let us focus on a single adversary case where $A\in V$ is the
node attacked by the adversary. Then we will 
generalize the idea for an arbitrary number of adversaries.

If \eqref{eq:SubsetRelation} is satisfied for all children of $u$, 
we know that node $u$ is 
not an adversary. If the relationship is not satisfied, that 
is \mbox{ $\Pi_v^{(u)} \nsubseteq \Pi_{P(u)}$} for at least 
one child of $u$, we consider node $u$ as a potential candidate
for being an adversary. For sure we know that
\[
\Pi_v^{(A)} \nsubseteq \Pi_{P(A)} \quad \forall v\in V: e_{Av} \in E,
\]
but depending on the subspace that the adversary reports, the 
relation \eqref{eq:SubsetRelation} may not be also satisfied 
for other nodes. Based on what the adversary reports there
would be two possible cases.

If the adversary pretends that it is a trustworthy node (just 
declares the received subspace from its parents) the above 
relation also fails for the children of $A$ who receive corrupted 
subspaces. On the other hand, if the adversary tells the truth and 
declares its corrupted subspace, we have
\[
\Pi_A^{(u)} \nsubseteq \Pi_{P(u)} \quad \forall u\in V: uA\in E.
\]
Thus the ambiguity set we have identified includes the adversary 
and its parents and/or its children depending on the adversary's 
report.

Repeating this procedure 
for every node in the network, we can identify sets of potential 
adversaries. 
We know that depending on the adversaries action there exists
ambiguity in finding their exact location. In fact in the worst
case, the uncertainty is within a set of nodes including the
adversary, its parents and its children. So if the distance between
adversaries is greater than two, the ``uncertainty'' sets do not
overlap. In this case we can easily distinguish between different
adversaries.
\end{proof} 

This procedure allows to identify adversaries (within the mentioned 
parent-children ambiguity set), even if one is in the shadow of another, and
even if we do not know their exact number, provided they are ``far
enough'' in the network to be distinguishable.




\section{Practical Implications for Topology Management}\label{chp:TopManagement}
In \S\ref{chp:TopInfer}, we demonstrated that using subspaces of all
nodes, we can infer the network topology under certain conditions. In
this section, we will show that even from what a
single node observes, it is possible to get some information regarding the
bottlenecks and clustering in the network. 

Leveraging this observation in the context of P2P networks, we
propose  algorithms that use this information in a distributed peer-initiated
manner to avoid bottlenecks and clustering.  

\subsection{Problem Statement and Motivation} \label{sec_mot}


In  peer-to-peer networks that employ network coding for content
distribution
(see for example Avalanche \cite{GkRo2005,GkMiRo2006}) 
we want to create and maintain a well-connected network topology, to allow the
information to flow fast between the nodes; 
however, this is not straightforward.
Peer-to-peer are very dynamically changing networks, where hundreds of
nodes may join and leave the network within seconds.  All nodes in this
network are connected to a small
number of neighbors (four to eight). An arriving node is allocated
neighbors among the active participating nodes\footnote{This is 
usually done by a central node which we call it (following Avalanche)
``registrat''. This is the central authority that keeps the list of all nodes
in the network and gives every new node a set of neighbors.},
which accept the solicited connection unless they have already reached
their maximum number of neighbors. As a result, nodes that arrive at
around the same time tend to get connected to each other, since they are
all simultaneously available and looking for neighbors. That is, we have
formation of clusters and bottlenecks in the network.

To avoid this problem, one method adopted in protocols is to
ask all nodes to periodically drop one neighbor and reconnect to a new one
among an active peers list.
This randomized rewiring  results in a fixed average
number of reconnections per node independently of how good or bad is
the formed network topology. Thus to achieve a good, on the average,
performance in terms of breaking clusters, it entails a much larger
number of rewiring  than required, and
unnecessary topology changes.

An alternative approach is to have peers initiate topology rewirings when
they detect they are in a cluster.
Clearly a central node 
could keep some structural information, \ie, keep track of the
current network topology, and use it to make more educated choices of
neighbor allocations.  However, the information this central node can
collect only reflects the \emph{overlay} network topology, and is
oblivious to bandwidth constraints from the underlying physical
links. Acquiring bandwidth information for the underlying physical
links at the central node requires costly estimation techniques over
large and heterogeneous networks, and steers towards a centralized
network operation. We will argue that such bottlenecks can be inferred
almost passively in a peer-initiated manner, thus alleviating these drawbacks.

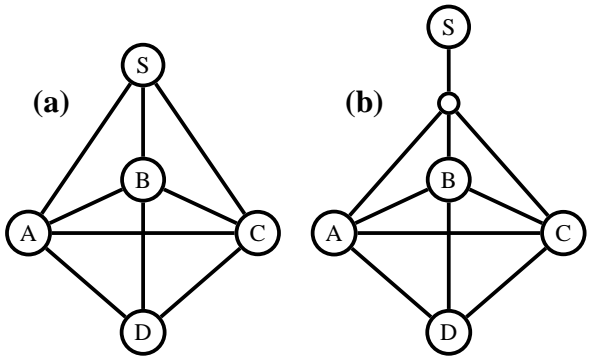
\begin{figure}[t]
\vspace{1cm}
\begin{center}
\psset{unit=0.04in}
\begin{pspicture}(-5,-5)(65,30)
\psset{linewidth=0.5mm}
\begin{small}
\rput(-2,25){\large \bf (a)}
\rput(10,30){\circlenode{S}{S}}
\rput(-5,8){\circlenode{A}{A}}
\rput(10,15){\circlenode{B}{B}}
\rput(25,8){\circlenode{C}{C}}
\rput(10,-5){\circlenode{D}{D}}
\ncline[linewidth=0.5mm,linecolor=black]{-}{S}{A}
\ncline[linewidth=0.5mm,linecolor=black]{-}{S}{B}
\ncline[linewidth=0.5mm,linecolor=black]{-}{S}{C}
\ncline[linewidth=0.5mm,linecolor=black]{-}{A}{B}
\ncline[linewidth=0.5mm,linecolor=black]{-}{A}{C}
\ncline[linewidth=0.5mm,linecolor=black]{-}{A}{D}
\ncline[linewidth=0.5mm,linecolor=black]{-}{B}{C}
\ncline[linewidth=0.5mm,linecolor=black]{-}{B}{D}
\ncline[linewidth=0.5mm,linecolor=black]{-}{C}{D}
%
\rput(39,25){\large \bf (b)}
\rput(50,35){\circlenode{S1}{S}}
\rput(50,25){\circlenode{S}{}}
\rput(35,8){\circlenode{A}{A}}
\rput(50,15){\circlenode{B}{B}}
\rput(65,8){\circlenode{C}{C}}
\rput(50,-5){\circlenode{D}{D}}
\ncline[linewidth=0.5mm,linecolor=black]{-}{S1}{S}
\ncline[linewidth=0.5mm,linecolor=black]{-}{S}{A}
\ncline[linewidth=0.5mm,linecolor=black]{-}{S}{B}
\ncline[linewidth=0.5mm,linecolor=black]{-}{S}{C}
\ncline[linewidth=0.5mm,linecolor=black]{-}{A}{B}
\ncline[linewidth=0.5mm,linecolor=black]{-}{A}{C}
\ncline[linewidth=0.5mm,linecolor=black]{-}{A}{D}
\ncline[linewidth=0.5mm,linecolor=black]{-}{B}{C}
\ncline[linewidth=0.5mm,linecolor=black]{-}{B}{D}
\ncline[linewidth=0.5mm,linecolor=black]{-}{C}{D}
\end{small}
\end{pspicture}
\end{center} 
\caption{The source $S$ distributes packets to the peers $A$, $B$, $C$
  and $D$ over the overlay network (a), that uses the underlying
  physical network (b).}
\label{fig_logical}
\vspace{-1.5em}
\end{figure}

Here, we will show that the coding vectors the
peers receive from their neighbors can be used to passively infer
bottleneck information. This allows individual nodes to initiate
topology changes to correct problematic connections. In particular,
peers by keeping track of the coding vectors they receive can detect
problems in both the overlay topology and the underlying physical
links. The following example illustrates these points.

\begin{example}
Consider the toy network depicted in Figure~\ref{fig_logical}(a) where
the edges correspond to logical (overlay network) links.  The source
$S$ has $n$ packets to distribute to four peers.  Nodes $A$, $B$ and
$C$ are directly connected to the source $S$, and also among
themselves with logical links, while node $D$ is connected to nodes
$A$, $B$ and $C$.  In this overlay network, there exist three
edge-disjoint paths between source and any other nodes.

Assume now (as shown in Figure~\ref{fig_logical}(b)) that the logical
links $SA$, $SB$, $SC$ share the bandwidth of the same underlying
physical link, which forms a bottleneck between the source and the
remaining nodes of the network.  As a result, assume the bandwidth on
each of these links is only $1/3$ of the bandwidth of the remaining
links. A central node (registrat), even 
if it keeps track of the complete logical
network structure by querying each node asking about its neighbors, is oblivious 
to the existence of the bottleneck and the asymmetry between the link bandwidths.


Node $D$ however, can infer this information by observing the coding
vectors it receives from its neighbors $A$, $B$ and $C$. Indeed, when
node $A$ receives a coded packet from the source, it will forward a
linear combination of the packets it has already collected to nodes
$B$ and $C$ and $D$. Now each of the nodes $B$ and $C$, once they
receive the packet from node $A$, they also attempt to send a coded
packet to node $D$. But these packets will not bring new information
to node $D$, because they will belong in the linear span of coding
vectors that node $D$ has already received. Similarly, when nodes $B$
and $C$ receive a new packet from the source, node $D$ will end up
being offered three coded packets, one from each of its neighbors, and
only one of the three will bring to node $D$ new information.
\hfill $\blacksquare$
\end{example}

More formally, the coding vectors nodes $A$, $B$ and $C$ will collect
will effectively span the same subspace; thus the coded packets they
will offer to node $D$ to download will belong in significantly
overlapping subspaces and will thus be redundant (we formalize these
intuitive arguments in
\S\ref{sec:TopManagement:TheoryFramework}). Node $D$ can infer from
this passively collected information that there is a bottleneck
between nodes $A$, $B$, $C$ and the source, and can thus initiate a
connection change.

\subsection{Theoretical Framework}\label{sec:TopManagement:TheoryFramework}
Here we use the same notations introduced in
\S\ref{chp:ProbDefMod}. For simplicity we will assume that the network
is synchronous\footnote{This is not essential for
the algorithms but simplifies the theoretical analysis.}. Nodes are
allowed to transmit linear combinations of their received packets only
at clock ticks, at a rate equal to the adjacent link bandwidth. 


Now we use the framework of \S\ref{chp:NotesVecSpace} to investigate the
information that we can obtain from the local information of a node's
subspace. From notations defined in \S\ref{chp:ProbDefMod}, we know that
for an arbitrary node $v$ we can write
\[
\Pi_v(t)=\sum_{i\in P(v)}\Pi^{(i)}_{v}(t).
\]
We are interested in understanding what information we can infer from
these received subspaces
$\Pi^{(i)}_{v}$, $i\in P(v)$, about bottlenecks in the 
network. For example, the overlap
of subspaces from the neighbors reveals some information about 
bottlenecks. Therefore, we need to show that such overlaps occur due to
topological properties and not due to particular random linear
combinations chosen by the network code. 

Let us assume that the subspaces
$\Pi^{(i)}_{v}$ a node $v$ receives from
its set of parents $P(v)$ have an intersection of dimension
$d$. Then we have the following observations.
\begin{observation}
The subspaces
$\Pi^{(i)}_{v}$, $i\in P(v)$, of the neighbors have an
intersection of size at least $d$ (see Corollary~\ref{cor:SpaceJointDim}).
\end{observation}

\begin{observation}
The min-cut between the set of nodes
$P(v)$ and the source is smaller than the min-cut between the node
$v$ and set $P(v)$ (see Theorem~\ref{thm:RandNetCodRate}).
\end{observation}
In the following, we will discuss algorithms that use such
observations for topology management.

\subsection{Algorithms} \label{sec_algo}
Our peer-initiated algorithms for topology management consist of three tasks:
\begin{enumerate}
\item Each peer decides whether it is satisfied with its connection or
  not, using a \emph{decision criterion}.
\item An unsatisfied peer sends a \emph{rewiring request}, that can
  contain different levels of information, either directly to the
  registrat, or to its neighbors (these are the only nodes the peer
  can communicate with).
\item Finally, the registrat, having received rewiring requests,
  \emph{allocates neighbors} to nodes to be reconnected.
\end{enumerate}

The decision criterion can capitalize on the fact that overlapping
received subspaces indicate an opportunity for improvement. 
For example, in the first algorithm we propose (Algorithm~1),
a node can decide it is not satisfied with a particular
neighbor, if it receives $k>0$, non-innovative coding vectors from it,
where $k$ is a parameter to be decided. Then it has each
unsatisfied node directly contact the registrat and specify the
neighbor it would like to change. The registrat randomly selects a new
neighbor. This algorithm, as we demonstrate through simulation
results, may lead to more rewirings than necessary: indeed, all nodes
inside a cluster may attempt to change their neighbors, while it would
have been sufficient for a fraction of them to do so.

Our second algorithm (Algorithm 2) uses a different decision criterion: for 
every two neighbors $u$ and $v$, each peer computes the rate at which 
the received joint space $\hat{\Pi}_u + \hat{\Pi}_v$ and intersection 
space $\hat{\Pi}_u\cap\hat{\Pi}_v$ increases. If the ratio between these 
two rates becomes greater than a threshold $\mathcal{T}$, the node decides 
it would like to change one of the two neighbors. However, instead of directly 
contacting the registrat, it uses a decentralized voting method that attempts 
to further reduce the number of reconnections. Then the registrat randomly 
selects and allocates one new neighbor for the nodes have sent rewiring request.

Our last proposed algorithm (Algorithm 3), while still peer-initiated
and decentralized, relies more than the two previous ones in the
computational capabilities of the registrat. The basic observation is
that, nodes in the same cluster will not only receive overlapping
subspaces from their parents, but moreover, they will end up
collecting subspaces with very small distance (this follows from
Theorem~\ref{thm:RandNetCodRate} and Corollary~\ref{cor:SpaceJointDim}
and is also illustrated through simulation results in
\S\ref{sec_simul}; see Figure~\ref{fig_test1}). Each unsatisfied peer $v$ sends a rewiring request
to the registrat, indicating to the registrat the subspace $\Pi_v$ it
has collected. A peer can decide it is not satisfied using for example
the same criterion as in Algorithm~2.

The registrat waits for a short time period, to collect requests from
a number of dissatisfied nodes. These are the nodes of the network
that have detected they are inside clusters. It then calculates the
distance between the identified subspaces to decide which peers belong
in the same cluster. While exact such calculations can be
computationally demanding, in practice, the registrat can use one of
the many hashing algorithms to efficiently do so. Finally the
registrat breaks the clusters by rewiring a small number of nodes in
each cluster. The allocated new neighbors are either nodes that belong
in different clusters, or, nodes that have not send a rewiring request
at all.

We will compare our proposed algorithms against the \emph{Random Rewiring}
currently employed by many peer-to-peer 
protocols (\eg, see \cite{GkRo2005,GkMiRo2006,HaLeBa-INRIA}). 
In this algorithm, each time a peer
receives a packet, with probability $p$ contacts the registrat and
asks to change a neighbor. The registrat randomly selects which
neighbor to change, and randomly allocates a new neighbor from the
active peer nodes.

\begin{figure}[htb]
\begin{center}
\includegraphics[width=2.3in,height=1.5in]{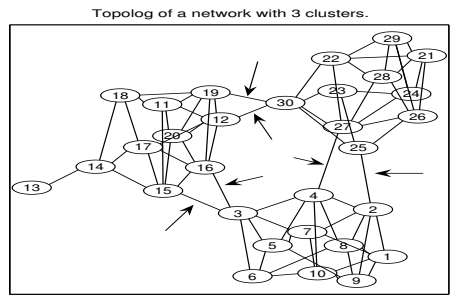}
\end{center}
\vspace{-1em}
\caption{A sample of topology with three clusters: cluster~$1$ contains nodes
  $1$--$10$, cluster~$2$ nodes $11$--$20$ and cluster~$3$ nodes
  $21$--$30$.}
\label{fig_topology}
\vspace{-1.5em}
\end{figure}

\begin{figure*}
\begin{center}
\includegraphics[width=3in,height=1.8in]{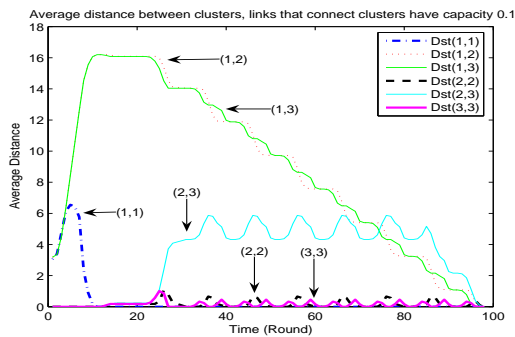}
\hspace{0.1\textwidth}
\includegraphics[width=3in,height=1.8in]{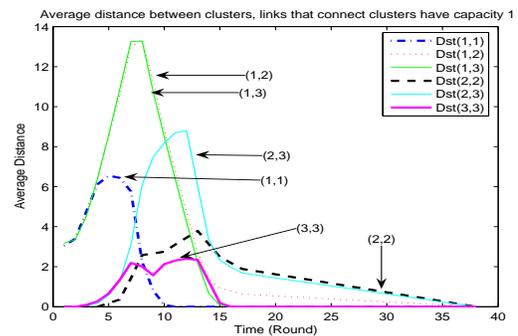}
\end{center}
\vspace{-1.5em}
\caption{Simulation results for the topology in Figure~\ref{fig_topology}, 
with bottleneck link capacity values equal to $0.1$ (left) and $1$ (right).}
\label{fig_test1}
\vspace{-.8em}
\end{figure*}

\subsection{Simulation Results} \label{sec_simul} For our simulation
results we will start from randomly generated topologies similar to 
Figure~\ref{fig_topology}, that consists of $30$ nodes connected into
three distinct clusters. The source is node~$1$, and belongs in the
first cluster. The bottleneck links are indicated with arrows (and
thus indicate the underlying physical link structure). Our first set
of simulation results depicted in Figure~\ref{fig_test1} show that the
subspaces within each cluster are very similar, while the subspaces
across clusters are significantly different, where we use the distance
measure $D_S(\cdot,\cdot)$ defined in (\ref{distance}). 
These results
indicate for example that knowledge of these subspaces will allow the
registrat to accurately detect and break clusters (Algorithm~3).

Our second set of simulation results considers again topologies with
three clusters: cluster~$1$ has $15$ nodes and contains the source,
cluster~$2$ has also $15$ nodes, while the number of nodes in
cluster~$3$ increases from $15$ to $250$.  During the simulations we
assume that the registrat keeps the nodes' degree between $2$ and $5$,
with an average degree of $3.5$. All edges correspond to unit capacity
links. 

We compare the performance of the three proposed algorithms in
\S\ref{sec_algo} with random rewiring. 
We implemented these algorithms as follows. For random
rewiring, every time a node receives a packet it changes one of its
neighbors with probability $p=\frac{8}{500}$. For Algorithm~1, we use
a parameter of $k=10$, and check whether the non-innovative packets
received exceed this value every four received packets. For
Algorithm~2, every node checks each received subspaces every four
received packets using the threshold value $\mathcal{T}=1$. 
Finally for Algorithm $3$, we assume that nodes use the same criterion
as in Algorithm 2 to decide whether they form part of a cluster, again
with $\mathcal{T}=1$. Dissatisfied nodes send their observed subspaces
to the registrat. The registrat assigns nodes $u$ and $v$ in the same
cluster if $d_S(\Pi_u,\Pi_v)\leq 7$.

Table~\ref{table_1} compares all algorithms with respect to the
average collection time, defined as the difference between the time a
peer receives the first packet and the time it can decode all packets,
and averaged over all peers. All algorithms perform similarly,
indicating that all algorithms result in breaking the clusters. It is
important to note that the average collection time is in terms of
number of exchanges needed and {\em does not} account for the delays
incurred due to rewiring. We compare the number of such rewirings
needed next.

Figure~\ref{fig_all} plots the average number of rewirings each
algorithm employs. Random rewiring incurs a number of rewirings
proportional to the number of P2P nodes, and independently from the
underlying network topology.
Our proposed algorithms on the other hand, adapt to the existence and
size of clusters. Algorithm~$3$ leads to the smallest number of
rewirings. Algorithm~$2$ leads to a larger number of rewirings, partly
due to that the new neighbors are chosen randomly and not in a manner
that necessarily breaks the clusters. The behavior of algorithm~$1$ is
interesting. This algorithm rewires any node that has received more
than $k$ non-innovative packets. Consider cluster $3$, whose size we
increase for the simulations. If $k$ is small with respect to the
cluster size, then a large number of nodes will collect close to $k$
non-innovative packets; thus a large number of nodes will ask for
rewirings. Moreover, even after rewirings that break the cluster
occur, some nodes will still collect linearly dependent information
and ask for additional rewirings. As cluster~$3$ increases in size,
the information disseminates more slowly within the cluster. Nodes in
the border, close to the bottleneck links, will now be the ones to
first ask for rewirings, long before other nodes in the network
collect a large number of non-innovative packets. Thus once the
clusters are broken, no new rewirings will be requested. This
desirable behavior of Algorithm~$1$ manifests itself for large
clusters; for small clusters, such as cluster~$2$, the second
algorithm for example achieves a better performance using less
reconnections.

\begin{figure}[htb]
\begin{center}
\includegraphics[width=3.3in]{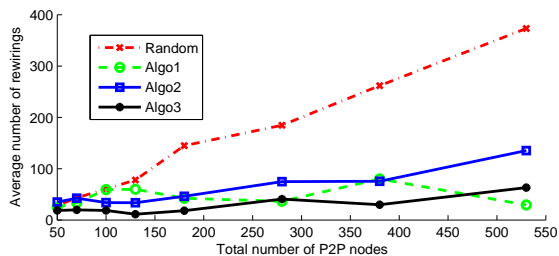}
\vspace{-0.5em}
\caption{Average number of rewirings, for a topology with three
  clusters: cluster~$1$ has $15$ nodes, cluster~$2$ has $15$ nodes,
  while the number of nodes in cluster~$3$ increases from $20$ to
  $250$ as described in Table~\ref{table_1}.}
\label{fig_all}
\end{center}
\vspace{-1.5em}
\end{figure}

\begin{table}[tbh!] 
\begin{center}
\caption{\label{table_1} Average Collection Time}  
\begin{tabular}{|c|c|c|c|c|}\hline
Topology & {Random} & {Algo 1} & {Algo 2} &{Algo 3}\\
\hline
$15$--$15$--$20$ & $20.98$ &  $22.14$  & $20.57$ &  $20.39$ \\
$15$--$15$--$40$ & $18.72$ &  $21.13$  & $19.36$ &  $19.47$ \\
$15$--$15$--$70$ & $18.88$ &  $21.54$  & $18.97$ &   $19.54$ \\
$15$--$15$--$100$ & $18.6$ &  $21.48$  & $18.91$ &  $21.42$ \\
$15$--$15$--$150$ & $19.56$ &  $20.85$  & $19.96$ &  $20.18$ \\
$15$--$15$--$250$ & $18.79$ &  $19.8$  & $19.18$ &  $18.99$ \\\hline
\end{tabular}
\end{center}
\vspace{-1.8em}
\end{table}

\section{Conclusions and Discussion}\label{chp:conclusion}

In this work we explored the properties of subspaces each node
collects in networks that employ randomized network coding and found
that there exists an intricate relationship between the structure of
the network and these properties. This observation led us to utilize
these relationships in several different applications.
As the first application, we studied the conditions under which we can  
passively infer the network topology during content distribution. We showed
that these conditions are not very restrictive and hold for a general
class of information dissemination protocols. 
As our second application, we focused on locating  Byzantine
attackers in the network. We studied and formulated this problem and found
that for the single adversary we can identify the adversary within an
uncertainty of two nodes. For the case of multiple adversaries, we
discussed a number of algorithms and conditions under which we can
guarantee identifiability. 
For our last application, we investigated the relation between the
bottlenecks in a logical network and the subspaces received at a specific network
node. We leveraged our observations to propose decentralized peer-initiated
algorithms for rewiring in P2P systems to avoid clustering in a
cost-efficient manner, and evaluated our algorithms through simulations results.

The applications studied in this paper demonstrate advantages of using
randomized network coding for network management and control, 
that are additional to throughput benefits. 
These are just a few examples and we
believe that there exist a lot more applications where we can use the
subspace properties developed in this work. We hope that these
properties will become part of a toolbox that can be used to develop
applications for systems that employ network coding techniques.


\appendices

\section{Proofs}\label{apndx:Proofs}

\begin{proof}[Proof of Lemma~\ref{lem:rand_choose_full_rank}]
First, let us fix a basis for $\Pi_S$. Then choosing $m$ vector 
uniformly at random from $\Pi_S$ is equivalent to choose an $m\times n$
matrix $\Bs{A}$ uniformly at random from $\Fbb_q$ and construct $\Pi=\sspan{\Bs{A}}$
with respect to this fixed basis.

It is well known (\eg, see \cite{LaTh-MathScand94}) that the number of different $m\times n$ 
matrices $\Bs{A}$ with rank $0\le k\le\min[m,n]$ over $\Fbb_q$ is equal to
\[
N_{m,n}(k)\triangleq q^{(m+n-k)k}\prod_{i=0}^{k-1} \frac{(1-q^{i-n})(1-q^{i-m})}{(1-q^{i-k})}.
\]
So we can write
\begin{align*}
\Prob{\dim(\Pi)=k} &= \frac{N_{m,n}(k)}{q^{mn}}.
\end{align*}
Then using the Taylor series 
$\frac{1}{1-\epsilon}=1+\epsilon+\epsilon^2+\cdots$ for $|\epsilon|<1$,
choosing $\epsilon=q^{-1}$, we can write
\begin{align*}
\Pr[\dim(\Pi)=k] = q^{-(m-k)(n-k)} [1-O(q^{-1})].
\end{align*}
By setting $k=\min[m,n]$ we are done.
\end{proof}

\begin{proof}[Proof of Lemma~\ref{lem:subset}]
The probability that all $m$ vectors are in the intersection is
\begin{eqnarray*}
\Prob{\Pi_1' \subset \Pi_2} &=& \left( \frac{q^{d_{12}}}{q^{d_1}} \right)^m = q^{(d_{12}-d_1)m},
\end{eqnarray*}
which is of order $\bigo{q^{-m}}$ provided that
$\Pi_1\nsubseteq\Pi_2$ , \ie, $d_{12}<d_1$.
\end{proof}

\begin{proof}[Proof of Lemma~\ref{lem:JointSpacePi_k}]
Let $\Bs{v}_1,\ldots,\Bs{v}_m$ be the vectors chosen randomly from $\Pi_S$
to construct $\Pi$, namely, we have $\Pi=\sspan{\Bs{v}_1,\dots,\Bs{v}_m}$.
Then construct the sequence of subspaces $\Pi(i)$, $i=0,\ldots,m$, as follows.
First, set $\Pi(0)\triangleq \Pi_k$ and then define $\Pi(i)$ for $i\neq 0$ 
recursively, $\Pi(i)=\Pi(i-1)+\sspan{\Bs{v}_i}$. We also define 
$d(i)\triangleq\dim(\Pi(i))$, $i=0,\ldots,m$. From Lemma~\ref{lem:subset}, by
choosing $\Pi_1=\Pi_S$, $\Pi_2=\Pi(i-1)$ and $m=1$ we deduce that 
$d(i)=d(i-1)+1$ with probability $1-\bigo{q^{-1}}$, unless $d(i-1)=n$.

Now we consider two cases. First, if $m+k\le n$ then we have $\dim(\Pi+\Pi_k)=k+m$ or
equivalently $\dim(\Pi\cap\Pi_k)=0$ with high probability, \ie, $1-\bigo{q^{-1}}$. 
Secondly, when $m+k>n$ we have $\dim(\Pi+\Pi_k)=n$ with probability $1-\bigo{q^{-1}}$.
From Lemma~\ref{lem:rand_choose_full_rank} we have $\dim(\Pi)=\min[m,n]$ w.h.p. 
So we have $\dim(\Pi\cap\Pi_k)=\dim(\Pi_k)+\dim(\Pi)-\dim(\Pi_k\cup\Pi)=k+\min[m,n]-n$. 

Combining these two cases we can write
\[
\dim(\Pi\cap\Pi_k)=(k+\min[m,n]-n)^+,
\]
w.h.p., which completes the proof.
\end{proof}

\begin{proof}[Proof of Corollary~\ref{cor:SpaceJointDim}]
Let us define $\Pi_{12}=\Pi_1\cap\Pi_2$, where
$d_{12}=\dim(\Pi_{12})$. Using Lemma~\ref{lem:JointSpacePi_k}, and
taking $\Pi_S=\Pi_1$ and $\Pi_k=\Pi_{12}$, we have
\begin{equation*}
\dim(\hat{\Pi}_1\cap\Pi_{12})=\min\left[d_{12},(m_1-(d_1-d_{12}))^+\right],
\end{equation*}
with probability $1-\bigo{q^{-1}}$.
Now, we can write
\begin{gather*}
\Prob{\hat{d}_{12}=\alpha} =\\ 
\Prob{\hat{d}_{12}=\alpha|\dim(\hat{\Pi}_1\cap\Pi_{12})=\beta} \Prob{\dim(\hat{\Pi}_1\cap\Pi_{12})=\beta}\\
+\Prob{\hat{d}_{12}=\alpha|\dim(\hat{\Pi}_1\cap\Pi_{12})\neq\beta} \Prob{\dim(\hat{\Pi}_1\cap\Pi_{12})\neq\beta},
\end{gather*}
where $\hat{d}_{12}=\dim(\hat{\Pi}_1\cap\hat{\Pi}_2)$. Substituting
$\beta=\min\left[d_{12},(m_1-(d_1-d_{12}))^+\right]$ we obtain
\begin{gather*}
\Prob{\hat{d}_{12}=\alpha} = \\
\Prob{\hat{d}_{12}=\alpha|\dim(\hat{\Pi}_1\cap\Pi_{12})=\beta} \left(1-\bigo{q^{-1}}\right) + \bigo{q^{-1}}.
\end{gather*}
Selecting $\alpha$ properly and using Lemma~\ref{lem:JointSpacePi_k}
one more time, we get
\begin{equation*}
\Prob{\hat{d}_{12}=\alpha} = 1-\bigo{q^{-1}},
\end{equation*}
where $\alpha=\min[\beta,(m_2-(d_2-\beta))^+]$, which completes the proof.
\end{proof}

\begin{proof}[Proof of Theorem~\ref{thm:SubSpaceGenPos}]
To prove the theorem, it is sufficient to show that
  (\ref{equ:genpos}) is valid for one specific $i$ with high probability. 
This is sufficient because
if $p_i$ is the probability that $\Pi$ is in general position with respect
  to each $\Pi_i$, $i=1,\ldots,r$, then the probability that $\Pi$ is in 
  general position with the whole family is lower bounded by
  $1-\sum_{i=1}^r (1-p_i)$. 

Now by applying Lemma~\ref{lem:JointSpacePi_k}, we know that 
  $p_i=1-\bigo{q^{-1}}$ which completes the proof.
\end{proof}

\begin{proof}[Proof of Lemma~\ref{lem:UpperBound_Ts}]
Here we assume that $n$ is very large. Then in Corollary~\ref{cor:SuffCondon_n}
we will derive a sufficient condition on the largeness of $n$.

Let $v$ be the node that has the longest path to the source $S$.
Because of Definition~\ref{def:waiting_time} we can write
$T_s\le \tau_v-1$. Then we may upper bound $\tau_v$ as follows
\[
\tau_v \le 2 + \max_{u\in P(v)} \tau_u,
\]
where $P(v)$ is the set of parents of $v$. Now we can repeat the
above argument until we reach the source $S$. So finally we have
\[
\tau_v \le 2 D(G),
\]
which leads to the lemma's assertion.
\end{proof}

\begin{proof}[Proof of Lemma~\ref{lem:ConsecutiveSubspaceSampling}]
Let us write
\begin{align*}
\dim(&\pi_u(1)\cap\Pi_v(j))\\ 
&\stackrel{(a)}{=} \dim\left(\pi_u(1)\cap (\Pi_v(j)\cap\Pi(0))\right)\\
&\stackrel{(b)}{=} \dim(\pi_u(1)\cap\pi_v(1))\\
&\stackrel{(c)}{=} \min[d_0,(k_u(1)+k_v(1)-d_0)^+,k_u(1),k_v(1)]\\
&= (k_u(1)+k_v(1)-d_0)^+\\
&< k_u(1),
\end{align*}
where $(a)$ follows because $\pi_u(1)\subseteq\Pi(0)$ and $(c)$ is a result of Corollary~2.
So $\forall j\in\{1,\ldots,t\}$ we have $\pi_u(1)\nsubseteq\Pi_v(j)$ which results in
$\Pi_u(i)\nsubseteq \Pi_v(j)$, $\forall i,j\in\{1,\ldots,t\}$. By symmetry, we have the 
second assertion of the lemma, namely, 
$\Pi_v(j)\nsubseteq \Pi_u(i)$, $\forall i,j\in\{1,\ldots,t\}$.

Now, it only remains to check $(b)$. We will prove this by induction. Obviously,
$\Pi(0)\cap\Pi_v(1)=\pi_v(1)$. Suppose that we have $\Pi(0)\cap\Pi_v(k)=\pi_v(1)$
where $k<t$ then we show that it also holds for $k+1$.

We know that $\pi_v(1)\subseteq\Pi(0)\cap\Pi_v(k+1)$. To show that 
$\Pi(0)\cap\Pi_v(k+1)\subseteq\pi_v(1)$ we proceed as follows. Let 
$w\in \Pi(0)\cap\Pi_v(k+1)$ then $w\in\Pi(0)$ and $w\in\Pi_v(k+1)=\sum_{i=1}^{k+1} \pi_v(i)$.
We may decompose $w$ as $w=\sum_{i=1}^{k+1} w_i$ where $w_i\in\pi_v(i)$.
Then we notice that $w_{k+1}=w-\sum_{i=1}^k w_i\in\Pi(k-1)$ and $\Pi(k-1)\cap\pi_v(k+1)=\emptyset$
w.h.p. (by Lemma~3). So we conclude that $w_{k+1}=0$ which means $w\in\Pi_v(k)$.
This shows that $w\in\Pi(0)\cap\Pi_v(k)$ where by induction assumption we have
$w\in\pi_v(1)$ and we are done.
\end{proof}


\begin{proof}[Proof of Corollary~\ref{cor:ConsecutiveSubspaceSampling}]
Because we have $\Pi_u(0)\nsubseteq\Pi_v(j)$ then by Lemma~2 we have 
$\pi_a(1)\nsubseteq\Pi_v(j)$ w.h.p. So as a result we have $\Pi_a(i)\nsubseteq\Pi_v(j-1)$
$\forall i,j\in\{1,\ldots,t\}$. Because $\Pi_b(j)\subseteq\Pi_v(j-1)$ we conclude that
$\Pi_a(i)\nsubseteq\Pi_b(j)$ $\forall i,j\in\{1,\ldots,t\}$ w.h.p. By symmetry, we also
deduce the other part of the corollary.
\end{proof}

\section{Algebraic Model for Synchronous Networks}\label{sec:Mod:AlgebraicModel}

In this appendix we employ an algebraic approach to analyze the dissemination 
protocol given in Algorithm~\ref{alg_diss}. This approach is similar to 
\cite{KoMe2003_TransNet} and \cite{HoKoMeEfShKa2006}, but differs in that 
we introduce memory into the coding process.

We introduce memory as follows. Suppose we are interested in finding the 
transfer function between the source and an arbitrary node $v$.
Let $\Bs{X}$ be a $n\times \ell$ matrix with  rows the $n$ packets (vectors) 
that the source wants to transmit to the receivers.  We assume that $\dim(\sspan{\Bs{X}})=n$.
Let  $\Bs{Y}(t)\in\mathbb{F}_q^{\xi\times \ell}$ be a matrix with rows 
the packets that pass through the  $\xi$ edges of the network at time $t$.  
Let $\Bs{Z}_v(t)$ be the set of packets that node $v$ receives.
Similarly to  \cite{KoMe2003_TransNet}, we will write state-space equations 
that involve these vectors; however, we will ensure that, at each time $t$, 
coding at each node occurs  across all the packets that the node has received 
before time $t$.

In each timeslot $t$, the source injects $|\Out(S)|$ packets into the network 
that are random linear combinations of the original source packets $\Bs{X}$. 
These linear combinations can be captured as $\Bs{M}(t)\Bs{X}$, 
where $\Bs{M}(t)\in\mathbb{F}_q^{|\Out(S)|\times n}$ is a random matrix. 
Intermediate network nodes will transmit packets on their outgoing edges 
depending on the network connectivity, and the state of the dissemination 
protocol.

The network connectivity can be captured by the  $\xi\times \xi$ adjacency 
matrix  $\mathcal{F}$  of the labeled line graph of the graph $G$, defined 
as follows
\begin{equation*}
\mathcal{F}_{ij}\triangleq\left\{ \begin{array}{ll}
1 & \head(e_i)=\tail(e_j),\\
0 & \textrm{otherwise.}
\end{array} \right.
\end{equation*}
To model random coding over a field  $\mathbb{F}_q$, we consider a sequence 
of random matrices $\Bs{F}_1^{(t)}, \ldots,\Bs{F}_{t-1}^{(t)}$  which conform 
to $\mathcal{F}$. That is, the entries of these matrices have for $i\ne j$ $(\Bs{F}_k^{(t)})_{ij}=0$ 
wherever $\mathcal{F}_{ij}=0$ and have random numbers from $\mathbb{F}_q$ in 
all other places. 

The dissemination protocol dictates when a node can start transmitting packets,
according to its waiting time (equivalently, when the outgoing edges of the 
node will have packets send through them).
To capture this, we will use the step function $u(t)$,
\[
u(t)\triangleq\left\{\begin{array}{l}
1\quad t\ge 0,\\
0\quad \textrm{otherwise},
\end{array}\right.
\]
and define the $\xi\times\xi$ diagonal matrix $\Bs{U}(t)$ as,
\[
\forall i\in E:\quad \Bs{U}_{ii}(t)\triangleq u\left(t-\tau_{\tail(i)}-1\right),
\]
where $\tau_v$ is the \emph{waiting time} for node $v$. In this section we
assume that the waiting times may have arbitrary values and we do not restrict them 
according to Definition~\ref{def:waiting_time}.

Using the above definitions, the set of packets (vectors) that each node $v$ 
receives in every time instant $t>0$ can be written as follows
\begin{equation}\label{eq:stateeq}
\left\{ \begin{array}{l}
\Bs{Y}(t) = \Bs{U}(t)\left(\Bs{A} \Bs{M}(t)\Bs{X}+\sum_{i=1}^{t-1} \Bs{F}_i^{(t)} \Bs{Y}(t-i)\right),\\
\\
\Bs{Z}_v(t) = \Bs{B}_v \Bs{Y}(t),
\end{array} \right.
\end{equation}
where $\Bs{Y}(0)=\Bs{0}$. In the above, $\Bs{A}\in\mathbb{F}_q^{\xi\times |\Out(S)|}$ is a 
matrix which represents the connection of node $S$ to the rest of the network. In the 
same way matrix $\Bs{B}_v\in\mathbb{F}_q^{|\In(v)|\times \xi}$ defines the connection 
of node $v$ to the set of edges in the network.

It is worth noting that although (\ref{eq:stateeq}) is written for the  packets transmitted 
on each edge, we can write the same set of equations for the coding vectors.

Suppose we are interested in finding the output of such a system at some time instant $T$. 
We can rewrite the above equations by defining new matrices as follows. We can collect 
the source random operations as
\begin{equation*}
\Bs{M}_T \triangleq \left[\begin{array}{c} 
\Bs{M}(1)\\
\vdots\\
\Bs{M}(T)
\end{array} \right] \in\mathbb{F}_q^{T|\Out(S)| \times n}.
\end{equation*}
For the states of system we define
\begin{equation*}
\Bs{Y}_T \triangleq \left[\begin{array}{c} 
\Bs{Y}(1)\\
\vdots\\
\Bs{Y}(T)
\end{array} \right] \in\mathbb{F}_q^{\xi T\times \ell}.
\end{equation*}
We also define a new set of matrices which represent the input-output relation. 
Using matrix $\Bs{A}$ we define the following matrix
\begin{equation*}
\Bs{A}_T \triangleq \Bs{I}_T\otimes \Bs{A}=\left[\begin{array}{ccc}
\Bs{A} &  & \\
 & \ddots & \\
 &  & \Bs{A}
\end{array} \right] \in\mathbb{F}_q^{\xi T\times T|\Out(S)|}.
\end{equation*}
For the connection of node $v$ we define
\begin{equation*}
\Bs{B}_v(T) \triangleq \left[\begin{array}{cc}
\Bs{0}_{|\In(v)|\times (T-1)\xi} & \Bs{B}_v \\
\end{array} \right] \in\mathbb{F}_q^{|\In(v)|\times \xi T}.
\end{equation*}
We define matrix $\Bs{F}_T$ which represent how the states are related to each other
\begin{equation*}
\boldsymbol{F}_T \triangleq \left[\begin{array}{ccccc} 
\boldsymbol{0} & \boldsymbol{0} & \boldsymbol{0} & \boldsymbol{0} &\cdots\\
\Bs{F}_1^{(2)} & \boldsymbol{0} & \boldsymbol{0} & \boldsymbol{0} & \cdots\\
\Bs{F}_2^{(3)} & \Bs{F}_1^{(3)} & \boldsymbol{0} & \boldsymbol{0} & \cdots\\
\Bs{F}_3^{(4)} & \Bs{F}_2^{(4)} & \Bs{F}_1^{(4)} & \boldsymbol{0} & \cdots\\
\vdots & \vdots & \vdots & \vdots & \ddots
\end{array}\right] \in\mathbb{F}_q^{\xi T\times \xi T}.
\end{equation*}
Finally, we use matrix $\Bs{U}_T$ that captures the time when transmissions start for each edge
\[
\Bs{U}_T \triangleq \left[\begin{array}{ccc}
 \Bs{U}(1) &  & \\
 & \ddots & \\
 & & \Bs{U}(T)
\end{array} \right] \in\mathbb{F}_q^{\xi T\times \xi T}.
\]
Using the above definitions, we can rewrite (\ref{eq:stateeq}) as follows
\begin{equation*}
\left\{ \begin{array}{l}
\Bs{Y}_T = \Bs{U}_T \left( \Bs{A}_T \Bs{M}_T \Bs{X} + \Bs{F}_T \Bs{Y}_T \right),\\
\\
\Bs{Z}_v(T) = \Bs{B}_v(T) \Bs{Y}_T.
\end{array} \right.
\end{equation*}
This equation can be solved to find the input-output transfer matrix at time $T$ 
which results in
\begin{equation}\label{eq:alg_input_output}
\boldsymbol{Z}_v(T)= \underbrace{\left[ \Bs{B}_v(T)(\Bs{I}-\Bs{U}_T\Bs{F}_T)^{-1} \Bs{U}_T\Bs{A}_T\Bs{M}_T\right]}_{\Bs{H}_{Sv}(T)} \Bs{X},
\end{equation}
where $\Bs{H}_{Sv}(T) \in \Fbb_{q}^{|\In(v)|\times n}$.
From the definition of matrix $\Bs{F}_T$, we know that it is a ``strictly lower triangular matrix'' 
which means $\Bs{F}_T$ is nilpotent and we have $\Bs{F}_T^T=0$. The same applies 
for the matrix $\Bs{U}_T\Bs{F}_T$, namely we have $(\Bs{U}_T\Bs{F}_T)^T=0$. 
So the matrix $(\Bs{I}-\Bs{U}_T\Bs{F}_T)^{-1}$ has an inverse which is equal to
\[
(\Bs{I}-\Bs{U}_T\Bs{F}_T)^{-1} = \left(\Bs{I}+\cdots+(\Bs{U}_T\Bs{F}_T)^{T-1}\right).
\]
Finally, note that if the nodes do not wait before starting the 
transmission ($\tau_v=0:\ \forall v\in V$), 
then we will have $\Bs{U}_T=\Bs{I}_{\xi T\times \xi T}$.

\subsection{Proof of Theorem~\ref{thm:RandNetCodRate}}\label{apndx:ProofThmRandNetCodeRate}

For simplicity, in the following proof, we assume that each edge of the network has 
capacity $1$. Edges with capacity more than $1$ can be modeled by replacing them with multiple 
edges of unit capacity.

From \eqref{eq:alg_input_output} the transfer matrix from $S$ to $v$ at time $T$ is equal to
$\Bs{H}_{Sv}(T)$. Knowing that the min-cut of node $v$ is $c_v$, we choose a set of $c_v$ incoming 
edges to $v$ such that there exist $c_v$ edge disjoint paths from $S$ to $v$ 
and find the input-output transfer matrix just for this set of edges. 
Then we can write
\begin{align}\label{eq:alg_input_output_mincut}
\Bs{\hat{H}}_{Sv}(T) &=  \Bs{\hat{B}}_v(T)(\Bs{I}-\Bs{U}_T\Bs{F}_T)^{-1} \Bs{U}_T\Bs{A}_T\Bs{M}_T \\
&=  \Bs{\hat{B}}_v(T) \left(\Bs{I}+\cdots+(\Bs{U}_T\Bs{F}_T)^{T-1}\right) \Bs{U}_T\Bs{A}_T\Bs{M}_T \nonumber, 
\end{align}
where $\Bs{\hat{H}}_{Sv}(T) \in \Fbb_{q}^{c_v\times n}$ and $\Bs{\hat{B}}_v(T)\in\Fbb_q^{c_v\times\xi T}$.
Let $f_{ij}^{(t,k)}$ denote 
for the entries of $\Bs{F}_k^{(t)}$ and $m_{ij}^{(t)}$ denote for the entries of $\Bs{M}(t)$. 
Every node in the network performs random linear network coding so $m_{ij}^{(t)}$ and 
$f_{ij}^{(t,k)}$ (those that are not zero) are chosen uniformly at random from $\mathbb{F}_q$.

From (\ref{eq:alg_input_output_mincut}) we know that each entry of ${\Bs{\hat{H}}}_{Sv}(T)$ 
is a polynomial of degree at most $T$ in variables $m_{ij}^{(t)}$ and $f_{ij}^{(t,k)}$. 
For $T>t_0(v)$ where $t_0(v)\triangleq\max_{i\in P(v)} \tau_i$, we 
know that there exists a trivial solution for variables $m_{ij}^{(t)}$ 
and $f_{ij}^{(t,k)}$ (which simply routes $c_v$ packets from $S$ to $v$ through the 
$c_v$ edge disjoint paths) that results in
\begin{equation}\label{eq:TrivialSolution-TransferMatrix}
\Bs{\hat{H}}_{Sv}(T) = \left[ \begin{array}{cc} \Bs{I}_{c_v} & \Bs{0}_{c_v\times (n-c_v)} \end{array} \right].
\end{equation}
Note that by changing the routing solution (in fact by changing the variables $m_{ij}^{(t)}$ 
properly) we could change the place of identity matrix in 
\eqref{eq:TrivialSolution-TransferMatrix} arbitrarily. We conclude that the
determinant of every $c_v\times c_v$ submatrix of $\Bs{\hat{H}}_{Sv}(T)$ (which 
is a polynomial of degree at most $c_v T$ in variables $m_{ij}^{(t)}$ and $f_{ij}^{(t,k)}$) is not
identical to zero. So by using the Schwartz-Zippel lemma \cite{MotRagh-RandAlgorithm} we can upper bound 
the probability that $\Bs{\hat{H}}_{Sv}(T)$ is not full rank if the variables
$m_{ij}^{(t)}$ and $f_{ij}^{(t,k)}$ are chosen uniformly at random as follows
\[
\Prob{\rank{\Bs{\hat{H}}_{Sv}(T)} < c_v} < \frac{c_v T}{q}.
\]
We can apply the same argument for $k< \frac{n}{c_v}$ consecutive timeslots to show that 
\[
\Prob{\rank{\Bs{\hat{H}}_{Sv}(T:T+k-1)} < kc_v} < \frac{k c_v (T+k)}{q},
\]
where 
\[
\Bs{\hat{H}}_{Sv}(T:T+k-1) \triangleq \left[\begin{array}{c} 
\Bs{\hat{H}}_{Sv}(T)\\
\vdots\\
\Bs{\hat{H}}_{Sv}(T+k-1)
\end{array} \right].
\]
Now let us define the event $\mc{A}_k(v)$ as follows
\[
\mc{A}_k(v):\ \rank{\Bs{\hat{H}}_{Sv}(T:T+k-1)} = kc_v.
\]
Then we can write
\begin{align*}
\Prob{\cap_{v\in V} \mc{A}_k(v)} &= 1- \Prob{\cup_{v\in V} \mc{A}_k^\complement(v)}\\
&\ge 1-\sum_{v\in V} \Prob{\mc{A}_k^\complement(v)}\\
&\ge 1- \frac{k(T+k)}{q}\sum_{v\in V} c_v,
\end{align*}
where $T>t_0$ and $t_0\triangleq\max_{v\in V} t_0(v)$.

This means that assuming $q$ is large enough we are sure that with high probability 
each node $v$ receives $c_v$ innovative packets per time slot for $t>t_0$.

\end{document}